\documentclass[11pt,letterpaper,reqno]{amsart}
\usepackage{geometry}
\usepackage[justification=centering]{caption}

\usepackage[authoryear,sort&compress]{natbib}
\bibpunct{[}{]}{;}{n}{,}{,}

\usepackage{amsmath}
\usepackage{amsfonts}
\usepackage{amssymb}
\usepackage{amsaddr}
\usepackage{comment}
\usepackage[shortlabels]{enumitem}
\usepackage{mathrsfs}
\usepackage{amsthm}
\usepackage{tikz}
    \usetikzlibrary{matrix}
    \usetikzlibrary{decorations.markings}
\usepackage{hyperref}
\usepackage[T1]{fontenc}
\usepackage{oldgerm}
\usepackage{scalerel,stackengine}
\usepackage{stmaryrd}
\usepackage{xcolor}
\usepackage{xfrac}
\usepackage{MnSymbol}

\newtheorem{theorem}{Theorem}[section]
\newtheorem{definition}{Definition}[section]

\newtheorem{proposition}{Proposition}[section]
\newtheorem{lemma}{Lemma}[section]

\newtheorem{remark}{Remark}[section]

\begin{document}

\title[Discrete Dirac reduction]{Discrete Dirac reduction of implicit Lagrangian systems with abelian symmetry groups}

\author{\'Alvaro Rodr\'iguez Abella}
\address{Instituto de Ciencias Matem\'aticas (CSIC-UAM-UC3M-UCM), Calle Nicol\'as Cabrera, 13-15, Madrid, 28049, Madrid, Spain.}
\email{alvrod06@ucm.es (corresponding author)}

\author{Melvin Leok}
\address{Department of Mathematics, UC San Diego, 9500 Gilman Drive, La Jolla, CA 92093-0112, USA.}
\email{mleok@ucsd.edu}


\begin{abstract}
This paper develops the theory of discrete Dirac reduction of discrete Lagrange--Dirac systems with an abelian symmetry group acting on the configuration space. We begin with the linear theory and, then, we extend it to the nonlinear setting using retraction compatible charts. We consider the reduction of both the discrete Dirac structure and the discrete Lagrange--Pontryagin principle, and show that they both lead to the same discrete Lagrange--Poincar\'e--Dirac equations. The coordinatization of the discrete reduced spaces relies on the notion of discrete connections on principal bundles. At last, we demonstrate the method obtained by applying it to a charged particle in a magnetic field, and to the double spherical pendulum.
\end{abstract}

\keywords{discrete mechanical systems, geometric numerical integration, Lagrange--Poincar\'e--Dirac equations, reduction by symmetries} 
\subjclass{37J39, 65P10, 70G65, 70H33}
\maketitle


\section{Introduction}

Symmetry reduction plays a central role in the field of geometric mechanics \cite{Ar1989, AbMa1978, MaRa1999}, and it involves expressing the dynamics of a mechanical system with symmetry in terms of the equivalence classes of group orbits on the space of solutions. This allows one to derive reduced equations of motion on a lower-dimensional reduced space which is obtained by quotienting the phase space by the symmetry action. The modern approach to symmetry reduction was introduced in \citet{Ar1966}, \citet{Sm1970}, \citet{Me1973}, and \citet{MaWe1974}, but the notion of symmetry reduction arises in earlier work of Lagrange, Poisson, Jacobi, and Noether.

Discrete variational mechanics~\cite{MaWe2001,west2006,LeZh2009} provides a discrete (in time) notion of Lagrangian dynamics, based on a discrete Hamilton's principle. This leads to discrete flow maps that are symplectic, and exhibit a discrete Noether's theorem. In turn, this naturally raises the question of whether one can develop a corresponding theory of symmetry reduction for discrete variational mechanics. Interest in this direction was motivated in part by attempts to understand the integrable discretization of the Euler top due to \citet{MoVe1991}. Prior work on discrete symmetry reduction includes a constrained variational formulation of discrete Euler--Poincar\'e reduction~\cite{MaPeSh1999, BoSu1999a}, the associated reduced discrete Poisson structure~\cite{MaPeSh2000}, discrete Euler--Poincar\'e reduction for field theories~\cite{Va2007} and discrete fluids~\cite{PaMuToKaMaDe2011,NaCo2017}, discrete Lie--Poisson integrators~\cite{MaRo2010,BaMa2018}, discrete higher-order Lagrange--Poincar\'e reduction~\cite{BlCoJi2019}, and a discrete notion of Routh reduction for abelian groups~\cite{jalnapurkar2005}. The resulting symplectic and Poisson integrators can be viewed as geometric structure-preserving numerical integrators, and this is an active area of study that is surveyed in \cite{HaLuWa2006}. In addition, discrete reduction theory can also be expressed in terms of composable groupoid sequences, which was the approach introduced in~\cite{We1996}, and explored further in \cite{MaDDMa2006}, and extended to field theories in~\cite{VaCa2007}.

The goal of this paper is to develop a discrete analogue of Dirac cotangent bundle reduction \cite{marsden2009} for Lagrangian systems, that is to say, a discrete Dirac reduction theory for implicit Lagrangian systems. The relevance of this type of reduction lies in the fact that it allows for treating a wide variety of systems with symmetry from a unifying viewpoint, including degenerate systems. Therefore, our discrete analogue allows us to construct geometric integrators for those systems. The theory presented here is limited to systems with abelian symmetry groups, but we expect to extend it to nonabelian groups in the future. Moreover, our work here is the first step in the construction of a category containing discrete Dirac structures that is closed under quotients, as proposed in \cite{leok2011}, which is fundamental for the study of discrete Dirac reduction by stages.

In addition, framing the discrete reduction theory in the context of discrete Dirac mechanics is significant, as it is the natural setting for studying discrete Hamiltonian mechanics on manifolds, particularly when expressed in terms of the discrete generalized energy. This is because the notion of discrete Hamiltonians~\cite{west2006,LeZh2009}, which are Type 2 and 3 generating functions, does not make intrinsic sense on a nonlinear manifold, since it is not possible to specify a covector on a nonlinear manifold without also specifying a base point. In contrast, the discrete generalized energy does make intrinsic sense, and is a more promising foundation on which to construct a discrete analogue of Hamiltonian mechanics on manifolds. Discrete Dirac mechanics is also the basis of a discrete theory of interconnections~\cite{PaLe2016}, which allows one to construct discretizations of complex multiphysics systems by interconnecting simpler subsystems, which provides a framework for geometric structure-preserving discretization of port-Hamiltonian systems~\cite{Va2006}.

Reduction theory can be addressed either in terms of the reduction of geometric structures like the symplectic or Poisson structures, or in terms of the reduction of variational principles. In this work, we will start with the discrete Lagrange--Dirac mechanics that was developed in~\cite{leok2011}, which can be viewed as a discrete analogue of Lagrange--Dirac mechanics, that can be formulated both in terms of Dirac structures~\cite{YoMa2006a} and the Hamilton--Pontryagin variational principle~\cite{YoMa2006b}. In order to coordinatize the reduced spaces arising in the reduction that we will perform, we rely on the notion of discrete principal connections that was introduced in \cite{leok2005} and further developed in \cite{fernandez2013}.

\subsection*{Overview} In Section~\ref{sec:Prelims}, we will recall the notion of discrete principal connection \cite{leok2005,fernandez2013} and discrete Dirac mechanics \cite{leok2010,leok2011}. Initially, we will consider the case where the configuration space $Q$ is a vector space, and the symmetry group $G$ is a vector subspace acting on $Q$ by addition, and show how the discrete Dirac structure is group-invariant, descends to a discrete Dirac structure on the quotient space in Section~\ref{sec:reductiondiracstrucure}. In Section~\ref{sec:reducedequations}, we will show how the reduced discrete Dirac structure can be used to derive the reduced discrete equations of motion, after we derive the reduced discrete Dirac differential, which is a map arising in the reduced discrete Tulczyjew triple. In Section~\ref{sec:reductionvariationalprinciple}, we also show that the discrete variational principle can be expressed in terms of a discrete generalized energy, and that if the discrete Lagrangian is group-invariant, then so is the discrete generalized energy. This induces a reduced discrete variational principle, which yields the same reduced discrete equations of motion. Then, in Section~\ref{sec:nonlinear}, we show that, in the nonlinear setting, we can use the notion of retractions~\cite{absil2008} and an atlas of retraction compatible charts~\cite{leok2011} to develop a global discrete theory, whose local representatives recover the vector space theory that we considered in the earlier part of the paper. This is significant, because that implies that with respect to an atlas of retraction compatible charts, the local theory that we initially constructed on each chart is compatible on overlapping domains, and induces a well-defined global theory on a manifold endowed with a semi-global retraction map. In Section~\ref{sec:examples}, we will illustrate the proposed discrete reduction theory by applying it to simulate the charged particle in a magnetic field, and the double spherical pendulum. Finally, we will summarize our contributions in the conclusion, and discuss future research directions.

\section{Preliminaries}\label{sec:Prelims}

\subsection{Group actions}

Let $Q$ be a smooth manifold, $G$ be a Lie group, and $\Phi^Q: G\times Q\to Q$ a free and proper left action of $G$ on $Q$. We will denote $\Phi^Q(g,q)=g\cdot q$ for each $q\in Q$ and $g\in G$. For a fixed $g\in G$, we define the map $\Phi_g^Q:Q\to Q$ as $q\mapsto q\cdot g$. Likewise, for a fixed $q\in Q$, we define the map $\Phi_q^Q:G\to Q$ as $g\mapsto q\cdot g$. The action being free and proper yields a principal bundle $\pi_{Q,\Sigma}: Q\to\Sigma$, where $\Sigma=Q/G$. We denote the equivalence class of $q\in Q$ by $[q]\in \Sigma$. In general, we use a superscript to denote the space where $G$ acts, and brackets $[\cdot]$ to denote the corresponding equivalence classes.

The diagonal action on $Q\times Q$ is given by
\begin{equation}\label{eq:diagonal action}
\Phi^{Q\times Q}: G\times Q\times Q\to Q\times Q,\qquad(g,(q_0,q_1))\mapsto(g\cdot q_0,g\cdot q_1).
\end{equation}
Likewise, the lift of the action to the cotangent bundle is given by
\begin{equation*}
\Phi^{T^*Q}: G\times T^*Q\to T^*Q,\qquad\left(g,p_q\right)\mapsto\left(\left(d\Phi_{g^{-1}}^Q\right)_{g\cdot q}^*(p_q)\right),
\end{equation*}
where the superscript $*$ denotes the adjoint map. In other words, for each $v_{g\cdot q}\in T_{g\cdot q} Q$ we have
\begin{equation*}
\left\langle g\cdot p_q,v_{g\cdot q}\right\rangle=\left\langle p_q,\left(d\Phi_{g^{-1}}^Q\right)_{g\cdot q}(v_{g\cdot q})\right\rangle.
\end{equation*}

The infinitesimal generators (or fundamental vector fields) of the action $\Phi^Q$ are denoted by $\xi_Q\in\mathfrak X(Q)$ for each $\xi\in\mathfrak g$, where $\mathfrak g$ is the Lie algebra of $G$, and analogously for the actions defined on the other spaces. In the same fashion, the momentum map $J: T^*Q\to\mathfrak g^*$ of the action is defined by $\langle J(p_q),\xi\rangle=\langle p_q,\xi_Q(q)\rangle$ for each $p_q\in T^*Q$ and $\xi\in\mathfrak g$.

\subsubsection{Group actions on vector spaces}\label{sec:vectoractions}

Now we suppose that $Q$ is a vector space, which allows us to identify 
\begin{equation}\label{eq:tangentidentification}
TQ=Q\times Q,\qquad T^*Q\simeq Q\times Q^*.    
\end{equation}
In addition, assume that $G\subset Q$ is a vector subspace\footnote{
Recall that any vector space can be regarded as a Lie group with the additive structure.} acting by addition, i.e., $g\cdot q=g+q$ for each $g\in G$ and $q\in Q$. In this case, we have $g^{-1}=-g$ and the action on $T^*Q$ reduces to,
\begin{equation}\label{eq:dualaction}
\Phi^{T^*Q}: G\times T^*Q\to T^*Q,\qquad(g,(q_0,p_1))\mapsto\left(g+q_0,p_1\right).
\end{equation}

To conclude, note that $\mathfrak g=G$ and $\mathfrak g^*=G^*$, and that the exponential map is the identity. Hence, $\xi_Q(q)=\xi$ for each $\xi\in\mathfrak g$ and $q\in Q$. Subsequently, the momentum map is given by $\langle J(q_0,p_0),\xi\rangle=\langle p_0,\xi\rangle$ for each $(q_0,p_0)\in T^*Q$ and $\xi\in\mathfrak g$.

\subsection{Discrete principal connections}\label{sec:discreteconnection}

Discrete principal connections were first introduced in \cite{leok2005} and further studied in \cite{fernandez2013}. Let $Q$ be a smooth manifold and $G$ be a Lie group acting freely and properly on $Q$.

\begin{definition}\label{def:discreteconnection}
A \emph{discrete principal connection} on the principal bundle $\pi_{Q,\Sigma}: Q\to\Sigma$ is a (smooth) function $\omega_d: Q\times Q\to G$ such that
\begin{enumerate}[(i)]
    \item $\omega_d(q_0,q_0)=e$, the identity element, for each $q_0\in Q$.
    \item $\omega_d(g_0\cdot q_0,g_1\cdot q_1)=g_1 \omega_d(q_0,q_1) g_0^{-1}$ for each $(q_0,q_1)\in Q\times Q$ and $g_0,g_1\in G$.
\end{enumerate}
\end{definition}

Discrete principal connections are not generally defined globally on $Q\times Q$, as its global existence implies a global trivialization of $TQ$. Instead, it is only semi-globally defined on a $G$-invariant open subset $U\subset Q\times Q$ containing the diagonal, i.e., $(q_0,q_0)\in U$ for each $q_0\in Q$. Nevertheless, in the following we assume that they are globally defined in order to simplify the notation.

A discrete principal connection $\omega_d$ enables us to define the \emph{discrete horizontal bundle} as
\begin{equation*}
H_d=\left\{(q_0,q_1)\in Q\times Q\mid\omega_d(q_0,q_1)=e\right\}\subset Q\times Q.
\end{equation*}
In addition, $\omega_d$ induces a \emph{discrete horizontal lift}, i.e., a map  $h_d: Q\times\Sigma\to Q\times Q$ that is the inverse of the diffeomorphism $(\textrm{\normalfont id}_Q\times\pi_{Q,\Sigma})|_{H_d}: H_d\to Q\times\Sigma$. We denote $\overline h_d=\pi_2\circ h_d: Q\times\Sigma\to Q$, where $\pi_2$ denotes the projection onto the second component.

\subsubsection{Local expression of discrete connections on vector spaces}\label{sec:localdiscreteconnection}

Suppose that $Q$ is a vector space and $G\subset Q$ is a vector subspace acting by addition. Therefore, $\Sigma=Q/G$ is a vector space and the projection $\pi_{Q,\Sigma}: Q\to\Sigma$ is linear. In this case, we assume that the horizontal lift, $h_d: Q\times\Sigma\to Q\times Q$, is a linear map (where the vector structure on the product is given by the direct sum). This, in turn, ensures that $\overline h_d: Q\times\Sigma\to Q$ is also linear.

Working in a trivialization of $\pi_{Q,\Sigma}$, i.e., supposing that $Q=\Sigma\times G$ with $\pi_{Q,\Sigma}=\pi_1$ and the vector structure given by the direct sum, then the horizontal lift of the discrete connection is determined by a map $\textrm{\normalfont h}_d: Q\times \Sigma\to G$ defined by $\textrm{\normalfont h}_d=\pi_2\circ\overline h_d$, i.e.,
\begin{equation*}
\overline h_d(q_0,x_1)=(x_1,\textrm{\normalfont h}_d(q_0,x_1)),\qquad q_0=(x_0,g_0)\in Q,\quad x_1\in \Sigma.
\end{equation*}
Of course, $\textrm{\normalfont h}_d$ is also linear and, from \cite[Equation (4.2)]{fernandez2013}, it satisfies
\begin{equation*}
\omega_d(q_0,q_1)=g_1-\textrm{\normalfont h}_d(q_0,x_1),\qquad q_0,q_1\in Q.
\end{equation*}
In particular, $\textrm{\normalfont h}_d(q_0,x_0)=g_0$ and $\textrm{\normalfont h}_d(g+q_0,x_1)=g+\textrm{\normalfont h}_d(q_0,x_1)$, for each $g\in G$ (to obtain this second equation we have used the equivariance of $\omega_d$).  

On the other hand, the dual of $Q$ is given by $Q^*=\Sigma^*\times G^*$, and the adjoint of $\textrm{\normalfont h}_d$ can be written as $\textrm{\normalfont h}_d^*=\left(\textrm{\normalfont h}_{d,Q}^*,\textrm{\normalfont h}_{d,\Sigma}^*\right): G^*\to Q^*\times\Sigma^*$ for some linear maps $\textrm{\normalfont h}_{d,Q}^*: G^*\to Q^*$ and $\textrm{\normalfont h}_{d,\Sigma}^*: G^*\to\Sigma^*$. They are related to the adjoint of $\overline h_d$ as follows
\begin{equation}\label{eq:adjointhorizontalliftlocal}
\overline h_d^*(p_0)=\left(\textrm{\normalfont h}_{d,Q}^*(r_0),w_0+\textrm{\normalfont h}_{d,\Sigma}^*(r_0)\right),\qquad p_0=(w_0,r_0)\in Q^*.
\end{equation}
Similarly, we denote by $\textrm{\normalfont h}_{d,Q}:\Sigma\to G$ the adjoint map of $\textrm{\normalfont h}_{d,Q}^*$, and analogously for $\textrm{\normalfont h}_{d,\Sigma}:\Sigma\to G$. In fact, we have $\textrm{\normalfont h}_{d,Q}(q_0)=\textrm{\normalfont h}_d(q_0,0)$ and $\textrm{\normalfont h}_{d,\Sigma}(x_1)=\textrm{\normalfont h}_d(0,x_1)$. Hence,
\begin{equation}\label{eq:decompositionhd}
\textrm{\normalfont h}_d(q_0,x_1)=\textrm{\normalfont h}_{d,Q}(q_0)+\textrm{\normalfont h}_{d,\Sigma}(x_1).
\end{equation}

For the sake of simplicity, we introduce the map $\textrm{\normalfont h}_d^0:\Sigma\times\Sigma\to G$, where $\textrm{\normalfont h}_d^0(x_0,x_1)=\textrm{\normalfont h}_d((x_0,0),x_1)$. Lastly, observe that $\textrm{\normalfont h}_d: Q\times\Sigma\to G$ may be regarded locally as a map defined on $(\Sigma\times G)\times\Sigma$. For this reason, we denote its partial derivatives by
\begin{equation*}
D_1 \textrm{\normalfont h}_d(q_0,x_1)=\frac{\partial\textrm{\normalfont h}_d}{\partial q_0}(q_0,x_1)=\left(\frac{\partial\textrm{\normalfont h}_d}{\partial x_0}(x_0,g_0,x_1),\frac{\partial\textrm{\normalfont h}_d}{\partial g_0}(x_0,g_0,x_1)\right),\quad D_2 \textrm{\normalfont h}_d(q_0,x_1)=\frac{\partial\textrm{\normalfont h}_d}{\partial x_1}(q_0,x_1),
\end{equation*}
for each $(q_0,x_1)=((x_0,g_0),x_1)\in(\Sigma\times G)\times\Sigma$. Due to the linearity of $\textrm{\normalfont h}_d$, they are given by
\begin{equation}\label{eq:partialderivativeshd}
D_1\textrm{\normalfont h}_d(q_0,x_1)(q)=\textrm{\normalfont h}_d(q,0),\quad D_2\textrm{\normalfont h}_d(q_0,x_1)(x)=\textrm{\normalfont h}_d(0,x),\qquad q\in Q,\quad x\in\Sigma.
\end{equation}

We conclude with the following straightforward result.

\begin{lemma}\label{lemma:Jlineartrivial}
Let $Q=\Sigma\times G$ be a trivialization of $\pi_{Q,\Sigma}$. Then for each $q_0=(x_0,g_0)\in Q$ and $p_0=(w_0,r_0)\in Q^*$ we have $J(q_0,p_0)=r_0$.
\end{lemma}

\subsection{Discrete Dirac mechanics}

Let us briefly recall the formulation of discrete Dirac mechanics introduced in \cite{leok2010,leok2011}. Here, we present only the unconstrained case. In the following, let $Q$ be a vector space.

\subsubsection{Generating functions and discrete Tulczyjew triple}

The Tulczyjew triple consists of three diffeomorphisms defined between the iterated tangent and cotangent bundles of a smooth manifold (cf. \cite{YoMa2006a}).
\begin{equation*}
	\begin{tikzpicture}
			\matrix (m) [matrix of math nodes,row sep=0.1em,column sep=3em,minimum width=2em]
			{	T^*(TQ) & T(T^*Q) & T^*(T^*Q)\\};
			\path[-stealth]
			(m-1-1) edge [bend left = 25] node [above] {$\gamma_Q$} (m-1-3)
			(m-1-2) edge [] node [above] {$\kappa_Q$} (m-1-1)
			(m-1-2) edge [] node [above] {$\Omega^\flat$} (m-1-3);
	\end{tikzpicture}
\end{equation*}
If $\Theta_{T^*(T^*Q)}\in\Omega^1(T^*(T^*Q))$ and $\Theta_{T^*(TQ)}\in\Omega^1(T^*(TQ))$ denote the tautological 1-forms on $T^*(T^*Q)$ and $T^*(TQ)$, respectively, then the Tulczyjew triple induce a symplectic form on $T(T^*Q)$,
\begin{equation*}
\Omega_{T(T^*Q)}=-d\left(\kappa_Q^*\Theta_{T^*(TQ)}\right)=d\left((\Omega^\flat)^*\Theta_{T^*(T^*Q)}\right)\in\Omega^2(T(T^*Q)).
\end{equation*}

The key idea of discrete Dirac mechanics is to build a discrete analogue of the Tulczyjew triple that retains the symplecticity of the maps involved. To that end, generating functions are utilized. 

By means of the Poincaré Lemma (see \cite[\S 3.2]{leok2011} for details), it can be shown that a map $F:T^*Q\to T^*Q$, given by $(q_0,p_0)\mapsto(q_1,p_1)$ is symplectic if and only if there exists a function $S_1:Q\times Q\to\mathbb R$, known as the \emph{Type 1 generating function}, such that $p_0=-D_1 S_1(q_0,q_1)$ and $p_1=D_2 S_1(q_0,q_1)$ for each $(q_0,q_1)\in Q\times Q$. We denote by $\iota_1:Q\times Q\to T^*Q\times T^*Q$ the map defined as $(q_0,q_1)\mapsto\big((q_0,p_0),(q_1,p_1)\big)$, and arrive at the following commutative diagram.
\begin{equation*}
\begin{array}{cc}
\begin{tikzpicture}
		\matrix (m) [matrix of math nodes,row sep=3em,column sep=1.5em,minimum width=2em]
		{	T^*Q\times T^*Q & & T^*(Q\times Q)\\
			 & Q\times Q & \\};
		\path[-stealth]
		(m-1-1) edge [] node [above] {$\kappa_Q^d$} (m-1-3)
		(m-2-2) edge [] node [below] {$\iota_1$} (m-1-1)
		(m-2-2) edge [] node [below] {$dS_1$} (m-1-3);
\end{tikzpicture} &
\begin{tikzpicture}
		\matrix (m) [matrix of math nodes,row sep=3em,column sep=1.5em,minimum width=2em]
		{	\big((q_0,p_0),(q_1,p_1)\big) & & (q_0,q_1,-p_0,p_1)\\
			 & (q_0,q_1) & \\};
		\path[-stealth]
		(m-1-1) edge [|->] node [above] {} (m-1-3)
		(m-2-2) edge [|->] node [left] {} (m-1-1)
		(m-2-2) edge [|->] node [above] {} (m-1-3);
\end{tikzpicture}
\end{array}
\end{equation*}
An analogous conclusion may be reached by using a function $S_2:Q\times Q^*\to\mathbb R$, known as \emph{Type 2 generating function}. In such case, we obtain a map
\begin{equation*}
\Omega_{d+}^\flat:T^*Q\times T^*Q\to T^*(Q\times Q^*),\qquad\big((q_0,p_0),(q_1,p_1)\big)\mapsto(q_0,p_1,p_0,q_1).
\end{equation*}
This way, we define the \emph{$(+)$-discrete Tulczyjew triple} as follows,
\begin{equation}\label{eq:Tulczyjew_d+}
	\begin{tikzpicture}[baseline=(current  bounding  box.center)]
			\matrix (m) [matrix of math nodes,row sep=0.1em,column sep=3em,minimum width=2em]
			{	T^*(Q\times Q) & T^*Q\times T^*Q & T^*(Q\times Q^*)\\
				(q_0,q_1,-p_0,p_1) & \big((q_0,p_0),(q_1,p_1)\big) & (q_0,p_1,p_0,q_1)\\};
			\path[-stealth]
			(m-1-1) edge [bend left = 25] node [above] {$\gamma_Q^{d+}$} (m-1-3)
			(m-1-2) edge [] node [above] {$\kappa_Q^d$} (m-1-1)
			(m-1-2) edge [] node [above] {$\Omega_{d+}^\flat$} (m-1-3)
			(m-2-2) edge [|->] node [] {} (m-2-1)
			(m-2-2) edge [|->] node [] {} (m-2-3);
	\end{tikzpicture}
\end{equation}
The discrete triple preserve the symplecticity of the continuous triple in the sense that the following is a natural symplectic form on $T^*Q\times T^*Q$,
\begin{equation*}
\Omega_{T^*Q\times T^*Q}=-d\left((\kappa_Q^d)^*\Theta_{T^*(Q\times Q)}\right)=d\left((\Omega_{d+}^\flat)^*\Theta_{T^*(Q\times Q^*)}\right)\in\Omega^2(T^*Q\times T^*Q).
\end{equation*}

On the other hand, a function $S_3:Q^*\times Q\to\mathbb R$, known as \emph{Type 3 generating function}, may be used to get a map
\begin{equation*}
\Omega_{d-}^\flat:T^*Q\times T^*Q\to T^*(Q^*\times Q),\qquad\big((q_0,p_0),(q_1,p_1)\big)\mapsto(p_0,q_1,-q_0,-p_1).
\end{equation*}
By using $\Omega_{d-}^\flat$ instead of $\Omega_{d+}^\flat$, we arrive to an analogous diagram called the \emph{$(-)$-discrete Tulczyjew triple}, which also inherits the symplecticity properties of the continuous triple.

\begin{equation}\label{eq:Tulczyjew_d-}
	\begin{tikzpicture}[baseline=(current  bounding  box.center)]
			\matrix (m) [matrix of math nodes,row sep=0.1em,column sep=3em,minimum width=2em]
			{	T^*(Q\times Q) & T^*Q\times T^*Q & T^*(Q^*\times Q)\\
				(q_0,q_1,-p_0,p_1) & \big((q_0,p_0),(q_1,p_1)\big) & (p_0,q_1,-q_0,-p_1)\\};
			\path[-stealth]
			(m-1-1) edge [bend left = 25] node [above] {$\gamma_Q^{d-}$} (m-1-3)
			(m-1-2) edge [] node [above] {$\kappa_Q^d$} (m-1-1)
			(m-1-2) edge [] node [above] {$\Omega_{d-}^\flat$} (m-1-3)
			(m-2-2) edge [|->] node [] {} (m-2-1)
			(m-2-2) edge [|->] node [] {} (m-2-3);
	\end{tikzpicture}
\end{equation}

\subsubsection{Discrete implicit Lagrangian systems}

\begin{definition}
The \emph{$(+)$-discrete induced Dirac structure} is defined as
\begin{equation*}
D^{d+} = \left\{(z_0,z_1,\alpha_{z^+})\mid z_0,z_1\in T^*Q,z^+=(q_0,p_1),\alpha_{z^+}=\Omega_{d+}^\flat(z_0,z_1)\right\}\subset (T^*Q\times T^*Q)\times T^*(Q\times Q^*).
\end{equation*}
where $z=(q,p)\in T^*Q\simeq Q\times Q^*$. Similarly, the \emph{$(-)$-discrete induced Dirac structure} is defined as
\begin{equation*}
D^{d-} = \left\{(z_0,z_1,\alpha_{z^-})\mid z_0,z_1\in T^*Q,z^-=(p_0,q_1),\alpha_{z^-}=\Omega_{d-}^\flat(z_0,z_1)\right\}\subset (T^*Q\times T^*Q)\times T^*(Q^*\times Q).
\end{equation*}
\end{definition}

\begin{remark}
A Dirac structure on a smooth manifold $M$ is a maximally isotropic subbundle of its Pontryagin bundle, $TM\oplus T^*M$ (see, for example, \cite{YoMa2006a}). Recall that the Whitney sum is a fibered sum over $M$, i.e., its elements are of the form $(v_x,\alpha_x)\in T_x M\times T_x^* M$ for each $x\in M$. Note that $\pi_1(z_0,z_1)=z_0\neq z^+=\pi_1(\alpha_{z^+})$, where $\pi_1$ denotes the projection onto the first component. Subsequently, $(z_0,z_1,\alpha_{z^+})$ are not elements of the Pontryagin bundle of $M=T^*Q$ and, hence, the $(+)$-discrete induced Dirac structure $D^{d+}$ defined above is not a Dirac structure on $T^*Q$ in the sense of Dirac structures on manifolds. The same holds for the $(-)$-discrete induced Dirac structure.
\end{remark}

Given a (possibly degenerate) discrete Lagrangian $L_d: Q\times Q\to\mathbb R$, its derivative is the map $dL_d: Q\times Q\to T^*(Q\times Q)$ given by
\begin{equation}\label{eq:dLdlocal}
dL_d(q_0,q_1)=(q_0,q_1,D_1 L_d(q_0,q_1),D_2 L_d(q_0,q_1)),\qquad(q_0,q_1)\in Q\times Q,
\end{equation}
where $D_i$ denotes the partial derivative with respect to the $i$-th component, $i=1,2$. The \emph{$(+)$-discrete Dirac differential} is the map
\begin{equation*}
\mathcal D^+ L_d=\gamma_Q^{d+}\circ dL_d: Q\times Q\longrightarrow T^*(Q\times Q^*).
\end{equation*}
Analogously, the \emph{$(-)$-discrete Dirac differential} is the map
\begin{equation*}
\mathcal D^- L_d=\gamma_Q^{d-}\circ dL_d: Q\times Q\longrightarrow T^*(Q^*\times Q).
\end{equation*}

At last, a \emph{discrete vector field} on $T^*Q$ is a sequence
\begin{equation*}
X_d=\left\{X_d^k=\big((q_k,p_k),(q_{k+1},p_{k+1})\big)\in T^*Q\times T^*Q\mid 0\leq k\leq N-1\right\}.
\end{equation*}

\begin{definition}[$(+)$-discrete implicit Lagrangian system]
A \emph{$(+)$-discrete implicit Lagrangian system}, also called a \emph{$(+)$-discrete Lagrange--Dirac system}, is a pair $(L_d,X_d)$, where $L_d$ is a discrete Lagrangian on $Q$ and $X_d$ is a discrete vector field on $T^*Q$, satisfying the \emph{$(+)$-discrete Lagrange--Dirac equations}, i.e.,
\begin{equation*}
\left(X_d^k,\mathcal D^+ L_d(q_k,q_k^+)\right)\in D^{d+},\qquad 0\leq k\leq N-1.
\end{equation*}
\end{definition}

The equations are locally given by
\begin{equation*}
q_k^+=q_{k+1},\quad p_{k+1}=D_2 L_d(q_k,q_k^+),\quad p_k=-D_1 L_d(q_k,q_k^+),\qquad 0\leq k\leq N-1.
\end{equation*}

\begin{definition}[$(-)$-discrete implicit Lagrangian system]
A \emph{$(-)$-discrete implicit Lagrangian system}, also called a \emph{$(-)$-discrete Lagrange--Dirac system}, is a pair $(L_d,X_d)$, where $L_d$ is a discrete Lagrangian on $Q$ and $X_d$ is a discrete vector field on $T^*Q$, satisfying the \emph{$(-)$-discrete Lagrange--Dirac equations}, i.e.,
\begin{equation*}
\left(X_d^k,\mathcal D^- L_d(q_{k+1}^-,q_{k+1})\right)\in D^{d-},\qquad 0\leq k\leq N-1.
\end{equation*}
\end{definition}

The equations are locally given by
\begin{equation*}
q_k=q_{k+1}^-,\quad p_{k+1}=D_2 L_d(q_{k+1}^-,q_{k+1}),\quad p_k=-D_1 L_d(q_{k+1}^-,q_{k+1}),\qquad 0\leq k\leq N-1.
\end{equation*}

Since we are considering the unconstrained case, both the $(+)$ and $(-)$-discrete Lagrange--Dirac equations are equivalent to the discrete Euler--Lagrange equations \cite{MaWe2001}.

\subsubsection{Variational structure for discrete implicit Lagrangian systems}

The discrete Lagrange--Dirac equations may also be obtained from a variational principle, as shown in \cite{leok2011}. As above, there exist two possible choices when performing discretization. Firstly, the \emph{$(+)$-discrete Pontryagin bundle} is the (vector) bundle over $Q$ given by
\begin{equation*}
(Q\times Q)\oplus_Q(Q\times Q^*)\simeq Q\times Q\times Q^*=\left\{(q_0,q_0^+,p_1)\mid q_0,q_0^+\in Q,~p_1\in Q^*\right\}.
\end{equation*}
Given a discrete Lagrangian $L_d: Q\times Q\to\mathbb R$, the \emph{$(+)$-discrete Lagrange--Pontryagin action} is the discrete augmented action defined as
\begin{equation}\label{eq:lagrangepontryaginaction}
\mathbb S_{L_d}^+\left[(q_k,q_k^+,p_{k+1})_{k=0}^N\right]=\sum_{k=0}^{N-1}\left(L_d(q_k,q_k^+)+\left\langle p_{k+1},q_{k+1}-q_k^+\right\rangle\right).
\end{equation}
The \emph{$(+)$-discrete Lagrange--Pontryagin principle},
\begin{equation}\label{eq:lagrangepontryaginprinciple}
\delta\mathbb S_{L_d}^+\left[(q_k,q_k^+,p_{k+1})_{k=0}^N\right]=0,
\end{equation}
is obtained by enforcing free variations $\left\{(\delta q_k,\delta q_k^+,\delta p_{k+1})\in Q\times Q\times Q^*\mid 0\leq k\leq N\right\}$ that vanish at the endpoints, i.e., $\delta q_0=\delta q_N=0$.

\begin{theorem}
The $(+)$-discrete Lagrange--Pontryagin principle is equivalent to the $(+)$-discrete Lagrange--Dirac equations.
\end{theorem}

On the other hand, the \emph{$(-)$-discrete Pontryagin bundle} is given by
\begin{equation*}
(Q\times Q)\oplus_Q(Q^*\times Q)\simeq Q\times Q^*\times Q=\left\{(q_1^-,p_0,q_1)\mid q_1^-,q_1\in Q,~p_0\in Q^*\right\}.
\end{equation*}
The \emph{$(-)$-discrete Lagrange--Pontryagin action} is defined as
\begin{equation}\label{eq:lagrangepontryaginactionminus}
\mathbb S_{L_d}^-\left[(q_{k+1}^-,p_k,q_{k+1})_{k=0}^N\right]=\sum_{k=0}^{N-1}\left(L_d(q_{k+1}^-,q_{k+1})+\left\langle p_k,q_k-q_{k+1}^-\right\rangle\right).
\end{equation}
The \emph{$(-)$-discrete Lagrange--Pontryagin principle},
\begin{equation}\label{eq:lagrangepontryaginprincipleminus}
\delta\mathbb S_{L_d}^-\left[(q_{k+1}^-,p_k,q_{k+1})_{k=0}^N\right]=0,
\end{equation}
is obtained by enforcing free variations $\left\{(\delta q_{k+1}^-,\delta p_k,\delta q_{k+1})\in Q\times Q\times Q^*\mid 0\leq k\leq N\right\}$ that vanish at the endpoints, i.e., $\delta q_0=\delta q_N=0$.

\begin{theorem}
The $(-)$-discrete Lagrange--Pontryagin principle is equivalent to the $(-)$-discrete Lagrange--Dirac equations.
\end{theorem}

\section{Reduction of the discrete Dirac structure}\label{sec:reductiondiracstrucure}

Let $Q$ be a vector space and $G\subset Q$ be a vector subspace acting by addition on $Q$. In this section we will show that the discrete induced Dirac structure on $Q$ is $G$-invariant and we will reduce it to the corresponding quotient.

\subsection{Trivializations of the tangent bundle and cotangent bundles}\label{sec:trivializationsTQ}

Let $\omega_d: Q\times Q\to G$ be a discrete connection form on $\pi_{Q,\Sigma}$. We define a right trivialization of $TQ=Q\times Q$ as 
\begin{equation*}
\lambda_d: Q\times Q\to Q\times(\Sigma\times G),\qquad(q_0,q_1)\mapsto\left(q_0,[q_1],\omega_d(q_0,q_1)\right).
\end{equation*}
Observe that it is a linear map with the vector structure of $Q\times Q$ given by the direct sum. Using \cite[Remark 4.3]{fernandez2013} it is straightforward to check that $\lambda_d$ is an isomorphism (of bundles over $T\Sigma=\Sigma\times\Sigma$) with inverse given by
\begin{equation*}
\lambda_d^{-1}: Q\times(\Sigma\times G)\to Q\times Q,\qquad(q_0,x_1,g_1)\mapsto\left(q_0,g_1+\overline h_d(q_0,x_1)\right).
\end{equation*}

The action of $G$ on $Q\times Q$ induces an action on $Q\times(\Sigma\times G)$ by means of $\lambda_d$. Using the equivariance of $\omega_d$, such an action is given by
\begin{equation*}
g\cdot(q_0,x_1,g_1)=(g+q_0,x_1,g_1),\qquad g\in G,\quad(q_0,x_1,g_1)\in Q\times(\Sigma\times G).
\end{equation*}
By construction, $\lambda_d$ is equivariant, i.e., $\lambda_d\circ\Phi_g^{Q\times Q}=\Phi_g^{Q\times(\Sigma\times G)}\circ\lambda_d$ for each $g\in G$. Subsequently, it descends to a left trivialization of $(Q\times Q)/G$, i.e., an isomorphism of the corresponding quotients, 
\begin{equation*}
[\lambda_d]: (Q\times Q)/G \longrightarrow (Q\times(\Sigma\times G))/G.
\end{equation*}
Furthermore, note that $(Q\times(\Sigma\times G))/G\simeq\Sigma\times(\Sigma\times G)$ via the isomorphism
\begin{equation}\label{eq:quotientQQ}
[q_0,x_1,g_1]\longmapsto\left([q_0],x_1,g_1\right).
\end{equation}

\begin{remark}\label{remark:lambdaminus}
For the $(-)$-case, it will be useful to trivialize the first factor instead, i.e., we consider the map
\begin{equation*}
\tilde\lambda_d:Q\times Q\to(\Sigma\times G)\times Q,\qquad(q_0,q_1)\mapsto\left([q_0],\omega_d(q_1,q_0),q_1\right).
\end{equation*}
Of course, all the computations performed for $\lambda_d$ are also valid for $\tilde\lambda_d$ by exchanging the order of the factors.
\end{remark}

On the other hand, we define a right trivialization of $T^*Q\simeq Q\times Q^*$ as
\begin{equation*}
\hat\lambda_d: Q\times Q^*\to Q\times(\Sigma^*\times\mathfrak g^*),\qquad(q_0,p_0)\mapsto\left(q_0,\pi_2\circ\overline h_d^*(p_0),J(q_0,p_0)\right).
\end{equation*}
Again, this trivialization is a linear map with the vector structure on $Q\times Q^*$ given by the direct sum.

\begin{remark}
Observe that the adjoint of the projection $\pi_{Q,\Sigma}: Q\to \Sigma$ yields a canonical embedding of $\Sigma^*$ into $Q^*$. However, there is not a canonical linear projection of $Q^*$ onto $\Sigma^*$. The discrete principal connection gives a choice of this projection via the adjoint of the horizontal lift, thus yielding the following split
\begin{equation*}
Q^*=\textrm{\normalfont im}(\pi_{Q,\Sigma}^*)\oplus\ker\left(\pi_2\circ\overline h_d^*\right).
\end{equation*}
In addition, $\langle\pi_{Q,\Sigma}^*(w),\xi_Q(q_0)\rangle=0$ for every $w\in\Sigma^*$ and $\xi\in\mathfrak g$, so we have
\begin{equation*}
\textrm{\normalfont im}(\pi_{Q,\Sigma}^*)=\{\xi_Q(q_0)\in Q\mid\xi\in\mathfrak g\}^0,
\end{equation*}
where the superscript $0$ denotes the annihilator.
\end{remark}

Using the previous remark, it can be seen that the inverse of $\hat\lambda_d$ is given by
\begin{equation*}
\hat\lambda_d^{-1}: Q\times(\Sigma^*\times\mathfrak g^*)\to Q\times Q^*,\qquad(q_0,w_0,\mu_0)\mapsto(q_0,\pi_{Q,\Sigma}^*(w_0)+(\mu_0)_Q(q_0)),
\end{equation*}
where $(\mu_0)_Q(q_0)\in Q^*$ is implicitly defined by the relations $\langle(\mu_0)_Q(q_0),\xi_Q(q_0)\rangle=\langle\mu_0,\xi\rangle$ for each $\xi\in\mathfrak g$ and $\left(\pi_2\circ\overline h_d^*\right)((\mu_0)_Q(q_0))=0$.

It is easy to check that the action of $G$ on $Q\times(\Sigma^*\times\mathfrak g^*)$ induced by $\hat\lambda_d$ is given by
\begin{equation*}
g\cdot(q_0,w_0,\mu_0)=(g+q_0,w_0,\mu_0),\qquad g\in G,\quad(q_0,w_0,\mu_0)\in Q\times(\Sigma^*\times\mathfrak g^*).
\end{equation*}

\begin{remark}\label{remark:hatlambdaminus}
For the $(-)$-case, it will be useful to trivialize the first factor instead. To that end, we consider the map
\begin{equation*}
\check\lambda_d:Q^*\times Q\to(\Sigma^*\times\mathfrak g^*)\times Q,\qquad(p_0,q_0)\mapsto\left(\pi_2\circ\overline h_d^*(p_0),J(q_0,p_0),q_0\right).
\end{equation*}
As for $\tilde\lambda_d$, all the computations performed for $\hat\lambda_d$ are also valid for $\check\lambda_d$ by exchanging the order of the factors.
\end{remark}

Since $Q\times Q$, $Q\times(\Sigma\times G)$, $Q\times Q^*$ and $Q\times(\Sigma^*\times\mathfrak g^*)$ are vector spaces, we may identify their tangent and cotangent bundles as in \eqref{eq:tangentidentification}, e.g., $T^*(Q\times Q)=Q\times Q\times Q^*\times Q^*$. Furthermore, we may consider the diagonal actions of $G$ on $T(Q\times Q)$ and $T(Q\times Q^*)$, as in \eqref{eq:diagonal action}. Likewise, on $T^*(Q\times Q)$ and $T^*(Q\times Q^*)$ we consider the cotangent lift of the action, as in \eqref{eq:dualaction}. In turn, these actions may be transferred to the corresponding (co)tangent bundles of $Q\times(\Sigma\times G)$ and $Q\times(\Sigma^*\times\mathfrak g^*)$ using the maps $\lambda_d$ and $\hat\lambda_d$, and their adjoint maps, accordingly. Note that since these maps are linear, their derivatives are the maps themselves. This can be done because $\lambda_d$ and $\hat\lambda_d$ are linear isomorphisms. For instance, the action of $G$ on $T^*(Q\times(\Sigma^*\times\mathfrak g^*))$ is induced from the action on $T^*(Q\times Q^*)$ using the map $\hat\lambda_d$ and its adjoint, as well as their inverses.

\subsection{Local expressions of the trivializations and quotients}\label{sec:quotientspaces}

In order to study the explicit expression of the maps introduced above, we choose a trivialization $Q=\Sigma\times G$ of $\pi_{Q,\Sigma}$ as in Section \ref{sec:localdiscreteconnection}. Using the local expression of the discrete connection, we have
\begin{equation}\label{eq:lambdadlocal}
\lambda_d(q_0,q_1)=(q_0,x_1,g_1-\textrm{\normalfont h}_d(q_0,x_1)),\qquad(q_0,q_1)\in Q\times Q.
\end{equation}
Hence,
\begin{equation}\label{eq:invlambdadlocal}
\lambda_d^{-1}(q_0,x_1,g_1)=(q_0,(x_1,g_1+\textrm{\normalfont h}_d(q_0,x_1))),\qquad(q_0,x_1,g_1)\in Q\times(\Sigma\times G).
\end{equation}
Its adjoint is given by
\begin{equation}\label{eq:adjointlambdadlocal}
\lambda_d^*(p_0,w_1,r_1)=\left(p_0-\langle r_1,\textrm{\normalfont h}_d(\cdot,0)\rangle,(w_1-\langle r_1,\textrm{\normalfont h}_d(0,\cdot)\rangle,r_1)\right),\quad(p_0,w_1,r_1)\in Q^*\times(\Sigma^*\times G^*).
\end{equation}
Hence,
\begin{equation}\label{eq:invadjointlambdadlocal}
\left(\lambda_d^*\right)^{-1}(p_0,p_1)=\left(p_0+\langle r_1,\textrm{\normalfont h}_d(\cdot,0)\rangle,w_1+\langle r_1,\textrm{\normalfont h}_d(0,\cdot)\rangle,r_1\right),\qquad (p_0,p_1)\in Q^*\times Q^*.
\end{equation}
Now we perform analogous computations for $\hat\lambda_d$, where we used \eqref{eq:adjointhorizontalliftlocal} and Lemma \ref{lemma:Jlineartrivial},
\begin{equation}\label{eq:hatlambdadlocal}
\hat\lambda_d(q_0,p_0)=\left(q_0,w_0+\textrm{\normalfont h}_{d,\Sigma}^*(r_0),r_0\right),\qquad(q_0,p_0)\in Q\times Q^*.
\end{equation}
The inverse is given by
\begin{equation}\label{eq:invhatlambdadlocal}
\hat\lambda_d^{-1}(q_0,w_0,\mu_0)=\big(q_0,(w_0-\textrm{\normalfont h}_{d,\Sigma}^*(\mu_0),\mu_0)\big),\qquad(q_0,w_0,\mu_0)\in Q\times(\Sigma^*\times\mathfrak g^*).
\end{equation}
Likewise, its adjoint is given by
\begin{equation}\label{eq:adjointhatlambdadlocal}
\hat\lambda_d^*(p_1,x_1,\xi_1)=\big(p_1,(x_1,\textrm{\normalfont h}_{d,\Sigma}(x_1)+\xi_1)\big),\qquad (p_1,x_1,\xi_1)\in Q^*\times(\Sigma\times\mathfrak g).
\end{equation}
At last, we have
\begin{equation}\label{eq:invadjointhatlambdadlocal}
\left(\hat\lambda_d^*\right)^{-1}(p_1,q_1)=\left(p_1,x_1,g_1-\textrm{\normalfont h}_{d,\Sigma}(x_1)\right),\qquad (p_1,q_1)\in Q^*\times Q.
\end{equation}

On the other hand, observe that the group action is locally given by
\begin{equation*}
g+q_0=(x_0,g+g_0),\qquad q_0=(x_0,g_0)\in Q,\quad g\in G,
\end{equation*}
where we identify $G\simeq\{0\}\times G\subset Q$. Using the local expression for the trivializations and their adjoints, we obtain the local expression for the actions of $G$ on the (co)tangent bundles of $Q\times(\Sigma\times G)$ and $Q\times(\Sigma^*\times\mathfrak g^*)$.

\begin{proposition}
The action of $G$ on $T^*(Q\times(\Sigma^*\times\mathfrak g^*))$ is locally given by
\begin{equation*}
g\cdot\big((q_0,w_0,\mu_0),(p_1,x_1,\xi_1)\big)=\big((g+q_0,w_0,\mu_0),(p_1,x_1,g+\xi_1)\big),
\end{equation*}
for each $g\in G$, $(q_0,w_0,\mu_0)\in Q\times(\Sigma^*\times\mathfrak g^*)$ and $(p_1,x_1,\xi_1)\in Q^*\times(\Sigma\times\mathfrak g)$. 

Analogously, the action on $T(Q\times(\Sigma^*\times\mathfrak g^*))$ is locally given by
\begin{equation*}
g\cdot\big((q_0,w_0,\mu_0),(q_1,w_1,\mu_1)\big)=\big((g+q_0,w_0,\mu_0),(g+q_1,w_1,\mu_1)\big),
\end{equation*}
for each $g\in G$ and $(q_0,w_0,\mu_0),(q_1,w_1,\mu_1)\in Q\times(\Sigma^*\times\mathfrak g^*)$.
\end{proposition}

\begin{proof}
As explained at the end of the previous section, the action of $G$ on $T^*(Q\times(\Sigma^*\times\mathfrak g^*))$ is induced from the action on $T^*(Q\times Q^*)$. Namely,
\begin{equation*}
\begin{tikzpicture}
\matrix (m) [matrix of math nodes,row sep=1em,column sep=3em,minimum width=2em]
{\big((q_0,w_0,\mu_0),(p_1,x_1,\xi_1)\big)\\
\Big(\big(q_0,(w_0-\textrm{\normalfont h}_{d,\Sigma}^*(\mu_0),\mu_0)\big),\big(p_1,(x_1,\textrm{\normalfont h}_{d,\Sigma}(x_1)+\xi_1)\big)\Big)\\
\Big(\big(g+q_0,(w_0-\textrm{\normalfont h}_{d,\Sigma}^*(\mu_0),\mu_0)\big),\big(p_1,(x_1,g+\textrm{\normalfont h}_{d,\Sigma}(x_1)+\xi_1)\big)\Big)\\
\Big(\big(g+q_0,w_0,\mu_0\big),\big(p_1,x_1,g+\xi_1\big)\Big).\\};
\path[-stealth]
(m-1-1) edge [|->] node [right] {\scriptsize\eqref{eq:invhatlambdadlocal},\;\eqref{eq:adjointhatlambdadlocal}} (m-2-1)
(m-2-1) edge [|->] node [right] {\scriptsize$G$ action} (m-3-1)
(m-3-1) edge [|->] node [right] {\scriptsize\eqref{eq:hatlambdadlocal},\;\eqref{eq:invadjointhatlambdadlocal}} (m-4-1);
\end{tikzpicture}
\end{equation*}
The computation for $T(Q\times(\Sigma^*\times\mathfrak g^*))$ is analogous.
\end{proof}

To conclude, by the previous proposition, we may define the following local isomorphisms,
\begin{equation}
\begin{array}{ccc}\label{eq:quotienttangentbundle}
T(Q\times(\Sigma^*\times\mathfrak g^*))/G & \simeq & \Sigma\times(\Sigma^*\times\mathfrak g^*)\times Q\times(\Sigma^*\times\mathfrak g^*),\\
\left[(q_0,w_0,\mu_0),(q_1,w_1,\mu_1)\right] & \mapsto & \left(x_0,w_0,\mu_0,-g_0+q_1,w_1,\mu_1\right),
\end{array}
\end{equation}
and
\begin{equation}\label{eq:quotientcotangentbundle}
\begin{array}{ccc}
T^*(Q\times(\Sigma^*\times\mathfrak g^*))/G & \simeq & \Sigma\times(\Sigma^*\times\mathfrak g^*)\times Q^*\times(\Sigma\times\mathfrak g),\\
\left[(q_0,w_0,\mu_0),(p_1,x_1,\xi_1)\right] & \mapsto & \left(x_0,w_0,\mu_0,p_1,x_1,-g_0+\xi_1\right).
\end{array}
\end{equation}

\begin{remark}\label{remark:isomorphismtangentcotangentminus}
By using $\check\lambda_d$ instead of $\hat\lambda_d$ (recall Remark \ref{remark:hatlambdaminus}), we arrive at analogous results. In particular, we have the isomorphisms
\begin{equation*}
T(Q\times(\Sigma^*\times\mathfrak g^*))/G\simeq Q\times(\Sigma^*\times\mathfrak g^*)\times\Sigma\times(\Sigma^*\times\mathfrak g^*),
\end{equation*}
and
\begin{equation*}
T^*((\Sigma^*\times\mathfrak g^*)\times Q)/G\simeq(\Sigma^*\times\mathfrak g^*)\times \Sigma\times(\Sigma\times\mathfrak g)\times Q^*.
\end{equation*}
\end{remark}

\subsection{Invariance of the discrete induced Dirac structure}

The definition of invariance for (continuum) Dirac structures \cite{marsden2009} can be extended to the discrete setting. More specifically, $D^{d+}$ is said to be \emph{$G$-invariant} if for each $g\in G$ and $(z_0,z_1,\alpha_{z^+})\in D^{d+}$ we have
\begin{equation*}
\left(\Phi_g^{T^*Q}(z_0),\Phi_g^{T^*Q}(z_1),\Phi_g^{T^*(Q\times Q^*)}(\alpha_{z+})\right)\in D^{d+}.
\end{equation*}
Similarly, $D^{d-}$ is said to be \emph{$G$-invariant} if for each $g\in G$ and $(z_0,z_1,\alpha_{z^-})\in D^{d-}$ we have
\begin{equation*}
\left(\Phi_g^{T^*Q}(z_0),\Phi_g^{T^*Q}(z_1),\Phi_g^{T^*(Q^*\times Q)}(\alpha_{z-})\right)\in D^{d-}.
\end{equation*}

\begin{proposition}
The discrete induced Dirac structures $D^{d+}$ and $D^{d-}$ are $G$-invariant.
\end{proposition}

\begin{proof}
Let $g\in G$ and $(z_0,z_1,\alpha_{z^+})\in D^{d+}$. By definition, we have
\begin{equation*}
\alpha_{z^+}=\Omega_{d+}^\flat(z_0,z_1)=(q_0,p_1,p_0,q_1).
\end{equation*}
Recall that the actions are given by $\Phi_g^{T^*Q}(q_0,p_0)=(g+q_0,p_0)$ and $\Phi_g^{T^*(Q\times Q^*)}(q_0,p_0,p_1,q_1)=(g+q_0,p_0,p_1,g+q_1)$, respectively.  Hence,
\begin{equation*}
\tilde z_i=\Phi_g^{T^*Q}(z_i)=\left(g+q_i,p_i\right),\qquad i=0,1,
\end{equation*}
and
\begin{equation*}
\tilde\alpha_{\tilde z^+}=\Phi_g^{T^*(Q\times Q^*)}(\alpha_{z+})=\left(g+q_0,p_1,p_0,g+q_1\right),
\end{equation*}
where $\tilde z^+=\left(g+q_0,p_1\right)$. It immediately follows that $\Omega_{d+}^\flat(\tilde z_0,\tilde z_1)=\tilde\alpha_{\tilde z^+}$, which establishes the desired result. The computation for $D^{d-}$ is analogous.
\end{proof}

On the other hand, the right trivialization $\hat\lambda_d$ allows $D^{d+}$ to induce a $(+)$-discrete Dirac structure on $Q\times(\Sigma^*\times\mathfrak g^*)$,
\begin{align*}
\hat D^{d+} & = \Big\{\left(\hat z_0,\hat z_1,\hat\alpha_{\hat z^+}\right)\mid z_0,z_1\in Q\times Q^*,\hat\alpha_{\hat z^+}\in T_{\hat z^+}^*(Q\times(\Sigma^*\times\mathfrak g^*)),\\
&\qquad\qquad\qquad\qquad z^+=(q_0,p_1),\left(z_0,z_1,\left(\hat\lambda_d^{-1},\hat\lambda_d^*\right)\left(\hat\alpha_{\hat z^+}\right)\right)\in D^{d+}\Big\},
\end{align*}
where, for the sake of simplicity, we let $\hat z=\hat\lambda_d(z)$ for each $z=(q,p)\in Q\times Q^*$. Equivalently, $\Omega_{d+}^\flat$ induces a map $\hat\Omega_{d+}^\flat$ between the trivialized spaces by imposing the commutativity of the following diagram,
\begin{equation*}
	\begin{tikzpicture}
			\matrix (m) [matrix of math nodes,row sep=5em,column sep=4em,minimum width=2em]
			{	(Q\times Q^*)\times(Q\times Q^*) & T^*(Q\times Q^*)\\
				(Q\times(\Sigma^*\times\mathfrak g^*))\times(Q\times(\Sigma^*\times\mathfrak g^*)) & T^*(Q\times(\Sigma^*\times\mathfrak g^*))\\};
			\path[-stealth]
			(m-1-1) edge [] node [left] {$\left(\hat\lambda_d,\hat\lambda_d\right)$} (m-2-1)
			(m-1-1) edge [] node [above] {$\Omega_{d+}^\flat$} (m-1-2)
			(m-1-2) edge [] node [right] {$\left(\hat\lambda_d,(\hat\lambda_d^*)^{-1}\right)$} (m-2-2)
			(m-2-1) edge [] node [above] {$\hat\Omega_{d+}^\flat$} (m-2-2);
	\end{tikzpicture}
\end{equation*}

This way, $\hat D^{d+}$ can be regarded as the $(+)$-discrete Dirac structure induced by $\hat\Omega_{d+}^\flat$, i.e.,
\begin{align*}
\hat D^{d+} & = \Big\{\left(\hat z_0,\hat z_1,\hat\alpha_{\hat z^+}\right)\mid\hat z_0,\hat z_1\in Q\times(\Sigma^*\times\mathfrak g^*),\hat\alpha_{\hat z^+}\in T_{\hat z^+}^*(Q\times(\Sigma^*\times\mathfrak g^*)),\vspace{0.1cm}\\
&\qquad\qquad\qquad\qquad\qquad \hat z^+=\hat\lambda_d(q_0,p_1),\hat\Omega_{d+}^\flat(\hat z_0,\hat z_1)=\hat\alpha_{\hat z^+}\Big\}.
\end{align*}

Of course, the $G$-invariance of $D^{d+}$ implies the $G$-invariance of $\hat D^{d+}$, since we have constructed the actions on the trivialized space so that $\hat\lambda_d$ is equivariant.

The same stands for the $(-)$-discrete Dirac structure. Namely, we have the following diagram,
\begin{equation*}
	\begin{tikzpicture}
			\matrix (m) [matrix of math nodes,row sep=5em,column sep=4em,minimum width=2em]
			{	(Q\times Q^*)\times(Q\times Q^*) & T^*(Q^*\times Q)\\
				(Q\times(\Sigma^*\times\mathfrak g^*))\times(Q\times(\Sigma^*\times\mathfrak g^*)) & T^*((\Sigma^*\times\mathfrak g^*)\times Q)\\};
			\path[-stealth]
			(m-1-1) edge [] node [left] {$\left(\hat\lambda_d,\hat\lambda_d\right)$} (m-2-1)
			(m-1-1) edge [] node [above] {$\Omega_{d-}^\flat$} (m-1-2)
			(m-1-2) edge [] node [right] {$\left(\check\lambda_d,(\check\lambda_d^*)^{-1}\right)$} (m-2-2)
			(m-2-1) edge [] node [above] {$\hat\Omega_{d-}^\flat$} (m-2-2);
	\end{tikzpicture}
\end{equation*}
Hence, we define
\begin{align*}
\hat D^{d-} & = \Big\{\left(\hat z_0,\hat z_1,\hat\alpha_{\hat z^-}\right)\mid\hat z_0,\hat z_1\in Q\times(\Sigma^*\times\mathfrak g^*),\hat\alpha_{\hat z^-}\in T_{\hat z^-}^*((\Sigma^*\times\mathfrak g^*)\times Q),\vspace{0.1cm}\\
& \qquad\qquad\qquad\qquad\qquad \hat z^-=\check\lambda_d(p_0,q_1),\hat\Omega_{d-}^\flat(\hat z_0,\hat z_1)=\hat\alpha_{\hat z^-}\Big\}.
\end{align*}

\subsection{Reduced discrete Dirac structure}

The $G$-invariance of $\hat D^{d+}$ ensures that it descends to a discrete Dirac structure on the quotient space. Since $\Omega_{d+}^\flat$ and $\hat\lambda_d$ are equivariant, so is $\hat\Omega_{d+}^\flat$, which induces a well-defined map between the quotient spaces,
\begin{equation*}
[\hat\Omega_{d+}^\flat]:\left[(Q\times(\Sigma^*\times\mathfrak g^*))\times(Q\times(\Sigma^*\times\mathfrak g^*))\right]/G\longrightarrow\left[T^*(Q\times(\Sigma^*\times\mathfrak g^*))\right]/G.
\end{equation*}
Then, the \emph{reduced $(+)$-discrete Dirac structure} is the structure induced by this map, i.e.,
\begin{multline}\label{eq:reduceddiracstructure}
[\hat D^{d+}] = \Big\{\left([\hat z_0,\hat z_1],[\hat\alpha_{\hat z^+}]\right)\mid[\hat z_0,\hat z_1]\in\left[(Q\times(\Sigma^*\times\mathfrak g^*))\times(Q\times(\Sigma^*\times\mathfrak g^*)\right]/G,\\
[\hat\alpha_{\hat z^+}]\in\left[T^*(Q\times(\Sigma^*\times\mathfrak g^*))\right]/G,[\hat\Omega_{d+}^\flat]\left([\hat z_0,\hat z_1]\right)=[\hat\alpha_{\hat z^+}]\Big\}.
\end{multline}
Locally, identifications \eqref{eq:quotienttangentbundle} and \eqref{eq:quotientcotangentbundle} enable us to regard $[\hat\Omega_{d+}^\flat]$ as a map between the trivializations,
\begin{equation*}
[\hat\Omega_{d+}^\flat]:\Sigma\times(\Sigma^*\times\mathfrak g^*)\times Q\times(\Sigma^*\times\mathfrak g^*)\longrightarrow\Sigma\times(\Sigma^*\times\mathfrak g^*)\times Q^*\times(\Sigma\times\mathfrak g).
\end{equation*}

\begin{lemma}
Working in a trivialization and using the above identification, we have
\begin{equation*}
[\hat\Omega_{d+}^\flat](x_0,w_0,\mu_0,q_1,w_1,\mu_1)=\Big(x_0,w_1,\mu_1,\big(w_0-\textrm{\normalfont h}_{d,\Sigma}^*(\mu_0),\mu_0\big),x_1,g_1-\textrm{\normalfont h}_{d,\Sigma}(x_1)\Big),
\end{equation*}
for each $(x_0,w_0,\mu_0,q_1,w_1,\mu_1)\in\Sigma\times(\Sigma^*\times\mathfrak g^*)\times Q\times(\Sigma^*\times\mathfrak g^*)$.
\end{lemma}

\begin{proof}
We employ the explicit local expression computed in the previous sections,
\begin{equation*}
\begin{tikzpicture}
\matrix (m) [matrix of math nodes,row sep=1em,column sep=3em,minimum width=2em]
{(x_0,w_0,\mu_0,q_1,w_1,\mu_1)\\
\Big(\big((x_0,0),w_0,\mu_0\big),\big(q_1,w_1,\mu_1\big)\Big)\\
\Big(\big((x_0,0),(w_0-\textrm{\normalfont h}_{d,\Sigma}^*(\mu_0),\mu_0)\big),\big(q_1,(w_1-\textrm{\normalfont h}_{d,\Sigma}^*(\mu_1),\mu_1)\big)\Big)\\
\Big(\big((x_0,0),(w_1-\textrm{\normalfont h}_{d,\Sigma}^*(\mu_1),\mu_1)\big),\big((w_0-\textrm{\normalfont h}_{d,\Sigma}^*(\mu_0),\mu_0),q_1\big)\Big)\\
\Big(\big((x_0,0),w_1,\mu_1\big),\big((w_0-\textrm{\normalfont h}_{d,\Sigma}^*(\mu_0),\mu_0),x_1,g_1-\textrm{\normalfont h}_{d,\Sigma}(x_1)\big)\Big)\\
\Big(x_0,w_1,\mu_1,\big(w_0-\textrm{\normalfont h}_{d,\Sigma}^*(\mu_0),\mu_0\big),x_1,g_1-\textrm{\normalfont h}_{d,\Sigma}(x_1)\Big).\\};
\path[-stealth]
(m-1-1) edge [|->] node [right] {\scriptsize\eqref{eq:quotienttangentbundle}} (m-2-1)
(m-2-1) edge [|->] node [right] {\scriptsize\eqref{eq:invhatlambdadlocal}} (m-3-1)
(m-3-1) edge [|->] node [right] {\scriptsize$\Omega_{d+}^\flat$ \eqref{eq:Tulczyjew_d+}} (m-4-1)
(m-4-1) edge [|->] node [right] {\scriptsize\eqref{eq:hatlambdadlocal},\;\eqref{eq:invadjointhatlambdadlocal}} (m-5-1)
(m-5-1) edge [|->] node [right] {\scriptsize\eqref{eq:quotientcotangentbundle}} (m-6-1);
\end{tikzpicture}
\end{equation*}
\end{proof}

Observe that $[\hat\Omega_{d+}^\flat]$ is a bundle morphism covering the identity if we regard the previous maps as bundles over $\Sigma$ with the projection onto the first component.

\begin{proposition}\label{prop:reduceddiractructurelocal}
Locally, the reduced $(+)$-discrete Dirac structure is given by
\begin{align*}
[\hat D^{d+}] & = \Big\{\left((x_0,w_0,\mu_0,q_1,w_1,\mu_1),(x_0,w_1,\mu_1,p,x_1,\xi)\right)\mid\\
&\qquad\qquad p=\big(w_0-\textrm{\normalfont h}_{d,\Sigma}^*(\mu_0),\mu_0\big),\xi=g_1-\textrm{\normalfont h}_{d,\Sigma}(x_1)\Big\}\vspace{0.1cm}\\
&\subset\big(\Sigma\times(\Sigma^*\times\mathfrak g^*)\times Q\times(\Sigma^*\times\mathfrak g^*)\big)\times\big(\Sigma\times(\Sigma^*\times\mathfrak g^*)\times Q^*\times(\Sigma\times\mathfrak g)\big).
\end{align*}
\end{proposition}

In the same fashion, we may define the \emph{reduced $(-)$-discrete Dirac structure} as
\begin{multline}\label{eq:reduceddiracstructureminus}
[\hat D^{d-}] = \Big\{\left([\hat z_0,\hat z_1],[\hat\alpha_{\hat z^-}]\right)\mid[\hat z_0,\hat z_1]\in\left[(Q\times(\Sigma^*\times\mathfrak g^*))\times(Q\times(\Sigma^*\times\mathfrak g^*))\right]/G,\\
[\hat\alpha_{\hat z^-}]\in\left[T^*((\Sigma^*\times\mathfrak g^*)\times Q)\right]/G,[\hat\Omega_{d-}^\flat]\left([\hat z_0,\hat z_1]\right)=[\hat\alpha_{\hat z^-}]\Big\}.
\end{multline}
By analogous computations using Remark \ref{remark:isomorphismtangentcotangentminus}, we arrive at the local expressions.

\begin{lemma}
Locally, we have
\begin{equation*}
[\hat\Omega_{d-}^\flat](q_0,w_0,\mu_0,x_1,w_1,\mu_1)=\Big(w_0,\mu_0,x_1,-x_0,-g_0+\textrm{\normalfont h}_{d,\Sigma}(x_0),(-w_1+\textrm{\normalfont h}_{d,\Sigma}^*(\mu_1),-\mu_1)\Big),
\end{equation*}
for each $(q_0,w_0,\mu_0,x_1,w_1,\mu_1)\in Q\times(\Sigma^*\times\mathfrak g^*)\times\Sigma\times(\Sigma^*\times\mathfrak g^*)$.
\end{lemma}

\begin{proposition}\label{prop:reduceddiractructurelocalminus}
Locally, the reduced $(-)$-discrete Dirac structure is given by
\begin{align*}
[\hat D^{d-}] & = \Big\{\left((q_0,w_0,\mu_0,x_1,w_1,\mu_1),(w_0,\mu_0,x_1,-x_0,\xi,p\right)\mid\\
&\qquad\qquad\xi=-g_0+\textrm{\normalfont h}_{d,\Sigma}(x_0),~p=\big(-w_1+\textrm{\normalfont h}_{d,\Sigma}^*(\mu_1),-\mu_1\big)\Big\}\vspace{0.1cm}\\
&\subset\big(Q\times(\Sigma^*\times\mathfrak g^*)\times\Sigma\times(\Sigma^*\times\mathfrak g^*)\big)\times\big((\Sigma^*\times\mathfrak g^*)\times \Sigma\times(\Sigma\times\mathfrak g)\times Q^*\big).
\end{align*}
\end{proposition}

\section{Discrete Lagrange--Poincar\'e--Dirac reduction}\label{sec:reducedequations}

Making use of the reduced discrete Dirac structure, we will compute the the reduced equations corresponding to a discrete Lagrange--Dirac system. Let $Q$ be a vector space and $G\subset Q$ be a vector subspace acting by addition on $Q$, and $L_d:Q\times Q\to\mathbb R$ be a (possibly degenerate) $G$-invariant discrete Lagrangian, i.e.,
\begin{equation*}
    L_d(g+q_0,g+q_1)=L_d(q_0,q_1),\qquad g\in G,\quad(q_0,q_1)\in Q\times Q.
\end{equation*}
The invariance of the discrete Lagrangian leads to the invariance of its partial derivatives.

\begin{lemma}\label{lemma:equivariancepartialderivatives}
If $L_d$ is $G$-invariant, then so are its partial derivatives $D_i L_d: Q\times Q\to Q^*$, $i=1,2$, i.e.,
\begin{equation*}
D_i L_d(g+q_0,g+q_1)=D_i L_d(q_0,q_1),\qquad(q_0,q_1)\in Q\times Q,\quad g\in G.
\end{equation*}
\end{lemma}

\begin{proof}
For each $q\in Q$ we have
\begin{align*}
\langle D_1 L_d(g+q_0,g+q_1),q\rangle & =  \displaystyle\lim_{h\to 0}\frac{L_d(g+q_0+hq,g+q_1)-L_d(g+q_0,g+q_1)}{h}\\
& =  \displaystyle\lim_{h\to 0}\frac{L_d(q_0+hq,q_1)-L_d(q_0,q_1)}{h}\\
& =  \langle D_1 L_d(q_0,q_1),q\rangle.
\end{align*}
An analogous computation establishes the result for $D_2 L_d$.
\end{proof}

\begin{remark}\label{remark:partialderivativesLd}
Working on a trivialization $Q=\Sigma\times G$ of $\pi_{Q,\Sigma}$, we may regard $L_d$ as a function defined on $(\Sigma\times G)\times(\Sigma\times G)$. This way, its partial derivatives can be written as
\begin{equation*}
D_1 L_d(q_0,q_1)=\frac{\partial L_d}{\partial q_0}(q_0,q_1)=\left(\frac{\partial L_d}{\partial x_0}(x_0,g_0,x_1,g_1),\frac{\partial L_d}{\partial g_0}(x_0,g_0,x_1,g_1)\right),
\end{equation*}
for each $(q_0,q_1)=(x_0,g_0,x_1,g_1)\in(\Sigma\times G)\times(\Sigma\times G)$, and analogous for $D_2 L_d(q_0,q_1)$.
\end{remark}

We may regard a discrete vector field, $X_{d+}$, that is a solution of the $(+)$-discrete Lagrange--Dirac equations, as a map
\begin{equation*}
\begin{array}{rccc}
X_{d+}: & Q\times Q & \longrightarrow & T(Q\times Q^*)=(Q\times Q^*)\times(Q\times Q^*),\\
& (q_0,q_1) & \longmapsto & \Big(\big(q_0,-D_1 L_d(q_0,q_1)\big),\big(q_1,D_2 L_d(q_0,q_1))\Big).
\end{array}
\end{equation*}
It follows from Lemma \ref{lemma:equivariancepartialderivatives} that this map is $G$-equivariant. Similarly to $\Omega_{d+}^\flat$, this induces a map between the trivialized spaces by imposing the commutativity of the following diagram,

\begin{equation*}
	\begin{tikzpicture}
			\matrix (m) [matrix of math nodes,row sep=5em,column sep=6em,minimum width=2em]
			{	Q\times Q & (Q\times Q^*)\times(Q\times Q^*)\\
				Q\times(\Sigma\times G) & (Q\times(\Sigma^*\times\mathfrak g^*))\times(Q\times(\Sigma^*\times\mathfrak g^*))\\};
			\path[-stealth]
			(m-1-1) edge [] node [left] {$\lambda_d$} (m-2-1)
			(m-1-1) edge [] node [above] {$X_{d+}$} (m-1-2)
			(m-1-2) edge [] node [right] {$\left(\hat\lambda_d,\hat\lambda_d\right)$} (m-2-2)
			(m-2-1) edge [] node [above] {$\hat X_{d+}$} (m-2-2);
	\end{tikzpicture}
\end{equation*}
Since $\hat X_{d+}$ is $G$-equivariant, it descends to a \emph{$(+)$-discrete reduced vector field}, 
\begin{equation}\label{eq:reducedvectorfield}
[\hat X_{d+}]:(Q\times(\Sigma\times G))/G\to T(Q\times(\Sigma^*\times\mathfrak g^*))/G.
\end{equation}
As above, locally we may regard it as $[\hat X_{d+}]:\Sigma\times(\Sigma\times G)\to\Sigma\times(\Sigma^*\times\mathfrak g^*)\times Q\times(\Sigma^*\times\mathfrak g^*)$.

Similarly, if $X_{d-}$ is a solution of the $(-)$-discrete Lagrange--Dirac equations, we may trivialize it as follows,
\begin{equation*}
	\begin{tikzpicture}
			\matrix (m) [matrix of math nodes,row sep=5em,column sep=6em,minimum width=2em]
			{	Q\times Q & (Q\times Q^*)\times(Q\times Q^*)\\
				(\Sigma\times G)\times Q & (Q\times(\Sigma^*\times\mathfrak g^*))\times(Q\times(\Sigma^*\times\mathfrak g^*))\\};
			\path[-stealth]
			(m-1-1) edge [] node [left] {$\tilde\lambda_d$} (m-2-1)
			(m-1-1) edge [] node [above] {$X_{d-}$} (m-1-2)
			(m-1-2) edge [] node [right] {$\left(\hat\lambda_d,\hat\lambda_d\right)$} (m-2-2)
			(m-2-1) edge [] node [above] {$\hat X_{d-}$} (m-2-2);
	\end{tikzpicture}
\end{equation*}
This way, we obtain the \emph{$(-)$-discrete reduced vector field},
\begin{equation}\label{eq:reducedvectorfieldminus}
[\hat X_{d-}]:((\Sigma\times G)\times Q)/G\to T(Q\times(\Sigma^*\times\mathfrak g^*))/G.
\end{equation}
Locally, we may regard it as a map $[\hat X_{d-}]:(\Sigma\times G)\times\Sigma\to Q\times(\Sigma^*\times\mathfrak g^*)\times\Sigma\times(\Sigma^*\times\mathfrak g^*)$.

\subsection{Reduced discrete Dirac differential}

The next step is to trivialize $\gamma_Q^{d+}$ and the exterior derivative of the discrete Lagrangian. To that end, we impose the commutativity of the following diagram,
\begin{equation*}
	\begin{tikzpicture}
			\matrix (m) [matrix of math nodes,row sep=5em,column sep=4em,minimum width=2em]
			{ Q\times Q & T^*(Q\times Q) & T^*(Q\times Q^*)\\
				 Q\times(\Sigma\times G) & T^*(Q\times(\Sigma\times G)) & T^*(Q\times(\Sigma^*\times\mathfrak g^*))\\};
			\path[-stealth]
			(m-1-2) edge [] node [left] {$\left(\lambda_d,(\lambda_d^*)^{-1}\right)$} (m-2-2)
			(m-1-2) edge [] node [above] {$\gamma_Q^{d+}$} (m-1-3)
			(m-1-3) edge [] node [right] {$\left(\hat\lambda_d,(\hat\lambda_d^*)^{-1}\right)$} (m-2-3)
			(m-2-2) edge [] node [above] {$\hat\gamma_Q^{d+}$} (m-2-3)
			(m-1-1) edge [] node [left] {$\lambda_d$} (m-2-1)
			(m-1-1) edge [] node [above] {$dL_d$} (m-1-2)
			(m-2-1) edge [] node [above] {$\hat{dL}_{d}$} (m-2-2);
	\end{tikzpicture}
\end{equation*}
Furthermore, since $\gamma_Q^{d+}$, $\lambda_d$ and $\hat\lambda_d$ are $G$-equivariant, so is $\hat\gamma_Q^{d+}$, by construction. The same holds for $\hat{dL}_{d}$, by Lemma \ref{lemma:equivariancepartialderivatives}. By composing the previous maps, we induce the $(+)$-discrete Dirac differential on the trivialized spaces,
\begin{equation*}
\hat{\mathcal D}^+ L_d=\hat\gamma_Q^{d+}\circ\hat{dL}_{d}: Q\times(\Sigma\times G)\longrightarrow T^*(Q\times(\Sigma^*\times\mathfrak g^*)).
\end{equation*}
Due to $G$-equivariance, it descends to the corresponding quotients, yielding the \emph{reduced $(+)$-discrete Dirac differential}, 
\begin{equation}\label{eq:reduceddiracdifferential}
[\hat{\mathcal D}^{+}L_d]:(Q\times(\Sigma\times G))/G\longrightarrow T^*(Q\times(\Sigma^*\times\mathfrak g^*))/G.
\end{equation}

\begin{proposition}\label{prop:reduceddiracdifferentiallocal}
Locally, the reduced $(+)$-discrete Dirac differential is given by a map $[\hat{\mathcal D}^{+}L_d]:\Sigma\times(\Sigma\times G)\to\Sigma\times(\Sigma^*\times\mathfrak g^*)\times Q^*\times(\Sigma\times\mathfrak g)$. Furthermore, for each $(x_0,x_1,g_1)\in\Sigma\times(\Sigma\times G)$, it is given by
\begin{equation*}
[\hat{\mathcal D}^{+}L_d](x_0,x_1,g_1)=\left(x_0,\frac{\partial L_d}{\partial x_1}+\textrm{\normalfont h}_{d,\Sigma}^*\left(\frac{\partial L_d}{\partial g_1}\right),\frac{\partial L_d}{\partial g_1},-\frac{\partial L_d}{\partial q_0},x_1,g_1+\textrm{\normalfont h}_{d,Q}(x_0,0)\right),
\end{equation*}
where the partial derivatives of $L_d$ are evaluated at $\left(x_0,0,x_1,g_1+\textrm{\normalfont h}_d^0(x_0,x_1)\right)$.
\end{proposition}

\begin{proof}
Once again, we use the explicit local expressions computed in the previous sections,
\begin{equation*}
\begin{tikzpicture}
\matrix (m) [matrix of math nodes,row sep=1em,column sep=3em,minimum width=2em]
{(x_0,x_1,g_1)\\
\big((x_0,0),x_1,g_1\big)\\
\Big((x_0,0),\big(x_1,g_1+\textrm{\normalfont h}_d^0(x_0,x_1)\big)\Big)\\
\Big((x_0,0),\big(x_1,g_1+\textrm{\normalfont h}_d^0(x_0,x_1)\big),D_1 L_d,D_2 L_d\Big)\\
\Big((x_0,0),D_2 L_d,-D_1 L_d,\big(x_1,g_1+\textrm{\normalfont h}_d^0(x_0,x_1)\big),\Big)\\
\displaystyle \left(\left((x_0,0),\frac{\partial L_d}{\partial x_1}+\textrm{\normalfont h}_{d,\Sigma}^*\left(\frac{\partial L_d}{\partial g_1}\right),\frac{\partial L_d}{\partial g_1}\right),\left(-\frac{\partial L_d}{\partial q_0},x_1,g_1+\textrm{\normalfont h}_d^0(x_0,x_1)-\textrm{\normalfont h}_{d,\Sigma}(x_1)\right)\right)\\
\displaystyle \left(x_0,\frac{\partial L_d}{\partial x_1}+\textrm{\normalfont h}_{d,\Sigma}^*\left(\frac{\partial L_d}{\partial g_1}\right),\frac{\partial L_d}{\partial g_1},-\frac{\partial L_d}{\partial q_0},x_1,g_1+\textrm{\normalfont h}_{d,Q}(x_0,0)\right),\\};
\path[-stealth]
(m-1-1) edge [|->] node [right] {\scriptsize\eqref{eq:quotientQQ}} (m-2-1)
(m-2-1) edge [|->] node [right] {\scriptsize\eqref{eq:invlambdadlocal}} (m-3-1)
(m-3-1) edge [|->] node [right] {\scriptsize\eqref{eq:dLdlocal}} (m-4-1)
(m-4-1) edge [|->] node [right] {\scriptsize$\gamma_Q^{d+}$ \eqref{eq:Tulczyjew_d+}} (m-5-1)
(m-5-1) edge [|->] node [right] {\scriptsize\eqref{eq:hatlambdadlocal},\;\eqref{eq:invadjointhatlambdadlocal}} (m-6-1)
(m-6-1) edge [|->] node [right] {\scriptsize\eqref{eq:quotientcotangentbundle}} (m-7-1);
\end{tikzpicture}
\end{equation*}
where the partial derivatives are evaluated at $\left(x_0,0,x_1,g_1+\textrm{\normalfont h}_d^0(x_0,x_1)\right)$ and we used \eqref{eq:decompositionhd} in the last step.
\end{proof}

For the $(-)$-case, we consider the diagram
\begin{equation*}
	\begin{tikzpicture}
			\matrix (m) [matrix of math nodes,row sep=5em,column sep=4em,minimum width=2em]
			{ Q\times Q & T^*(Q\times Q) & T^*(Q^*\times Q)\\
				 (\Sigma\times G)\times Q & T^*((\Sigma\times G)\times Q) & T^*((\Sigma^*\times\mathfrak g^*)\times Q)\\};
			\path[-stealth]
			(m-1-2) edge [] node [left] {$\left(\tilde\lambda_d,(\tilde\lambda_d^*)^{-1}\right)$} (m-2-2)
			(m-1-2) edge [] node [above] {$\gamma_Q^{d-}$} (m-1-3)
			(m-1-3) edge [] node [right] {$\left(\check\lambda_d,(\check\lambda_d^*)^{-1}\right)$} (m-2-3)
			(m-2-2) edge [] node [above] {$\hat\gamma_Q^{d-}$} (m-2-3)
			(m-1-1) edge [] node [left] {$\tilde\lambda_d$} (m-2-1)
			(m-1-1) edge [] node [above] {$dL_d$} (m-1-2)
			(m-2-1) edge [] node [above] {$\hat{dL}_{d}$} (m-2-2);
	\end{tikzpicture}
\end{equation*}
Subsequently, we arrive at the \emph{reduced $(-)$-discrete Dirac differential}, 
\begin{equation}\label{eq:reduceddiracdifferentialminus}
[\hat{\mathcal D}^{-}L_d]:((\Sigma\times G)\times Q)/G\longrightarrow T^*((\Sigma^*\times\mathfrak g^*)\times Q)/G.
\end{equation}

\begin{proposition}\label{prop:reduceddiracdifferentiallocalminus}
Locally, the reduced $(-)$-discrete Dirac differential is given by a map $[\hat{\mathcal D}^{-}L_d]:(\Sigma\times G)\times\Sigma\to(\Sigma^*\times\mathfrak g^*)\times\Sigma\times(\Sigma\times\mathfrak g)\times Q^*$. Furthermore, for each $(x_0,g_0,x_1)\in(\Sigma\times G)\times\Sigma$, it is given by
\begin{equation*}
[\hat{\mathcal D}^{-}L_d](x_0,g_0,x_1)=\left(-\frac{\partial L_d}{\partial x_0}-\textrm{\normalfont h}_{d,\Sigma}^*\left(\frac{\partial L_d}{\partial g_0}\right),-\frac{\partial L_d}{\partial g_0},x_1,-x_0,-g_0-\textrm{\normalfont h}_{d,Q}(x_1,0),-\frac{\partial L_d}{\partial q_1}\right),
\end{equation*}
where the partial derivatives of $L_d$ are evaluated at $\left(x_0,g_0+\textrm{\normalfont h}_d^0(x_1,x_0),x_1,0\right)$.
\end{proposition}

\subsection{\texorpdfstring{$(+)$}{(+)}-discrete Lagrange--Poincar\'e--Dirac equations}

Define $\hat L_{d+}: Q\times(\Sigma\times G)\to\mathbb R$ by the condition $\hat L_{d+}\circ\lambda_d=L_d$. Clearly, it is $G$-invariant, thus inducing the  \emph{$(+)$-discrete reduced Lagrangian}, 
\begin{equation*}
l_{d+}:(Q\times(\Sigma\times G))/G\longrightarrow\mathbb R.
\end{equation*}
Locally, we may regard it as $l_{d+}:\Sigma\times\Sigma\times G\to\mathbb R$. It is easy to check that
\begin{equation}\label{eq:reducedlagrangianlocal}
l_{d+}(x_0,x_1,g_1)=L_d\big(x_0,0,x_1,g_1+\textrm{\normalfont h}_d^0(x_0,x_1)\big),\qquad(x_0,x_1,g_1)\in\Sigma\times\Sigma\times G.
\end{equation}

\begin{lemma}\label{lemma:partialderivativesld}
Locally, for each $(x_0,x_1,g_1)\in\Sigma\times\Sigma\times G$ we have 
\begin{align*}
\displaystyle \frac{\partial L_d}{\partial q_0}&=\left(\frac{\partial l_{d+}}{\partial x_0}-\left\langle\frac{\partial l_{d+}}{\partial g_1},\textrm{\normalfont h}_d^0(\cdot,0)\right\rangle,-\left\langle\frac{\partial l_{d+}}{\partial g_1},\textrm{\normalfont h}_d((0,\cdot),0)\right\rangle\right),\vspace{0.1cm}\\
\displaystyle \frac{\partial L_d}{\partial x_1}&=\frac{\partial l_{d+}}{\partial x_1}-\left\langle\frac{\partial l_{d+}}{\partial g_1},\textrm{\normalfont h}_d^0(0,\cdot)\right\rangle,\qquad\frac{\partial L_d}{\partial g_1}=\frac{\partial l_{d+}}{\partial g_1},
\end{align*}
where the partial derivatives of $L_d$ and $l_{d+}$ are evaluated at $\left(x_0,0,x_1,g_1+\textrm{\normalfont h}_d^0(x_0,x_1)\right)$ and $(x_0,x_1,g_1)$, respectively.
\end{lemma}

\begin{proof}
Let $(x_0,g_0,x_1,g_1')\in(\Sigma\times G)\times(\Sigma\times G)$. Using \eqref{eq:lambdadlocal}, we obtain
\begin{align*}
L_d\big((x_0,g_0),(x_1,g_1')\big) & = \left(\hat L_{d+}\circ\lambda_d\right)\big((x_0,g_0),(x_1,g_1')\big)\vspace{0.1cm}\\
& =  \hat L_{d+}\left((x_0,g_0),x_1,g_1'-\textrm{\normalfont h}_d\big((x_0,g_0),x_1\big)\right)\vspace{0.1cm}\\
& = l_{d+}\left(x_0,x_1,g_1'-\textrm{\normalfont h}_d\big((x_0,g_0),x_1\big)\right).
\end{align*}
Recall that the partial derivatives of $\textrm{\normalfont h}_d$ are given by \eqref{eq:partialderivativeshd}. Using this fact, as well as the chain rule and the relation above, we obtain
\begin{alignat*}{3}
\displaystyle\frac{\partial L_d}{\partial x_0}&=\frac{\partial l_{d+}}{\partial x_0}-\left\langle\frac{\partial l_{d+}}{\partial g_1},\textrm{\normalfont h}_d^0(\cdot,0)\right\rangle,&\qquad  
\displaystyle\frac{\partial L_d}{\partial g_0}&=-\left\langle\frac{\partial l_{d+}}{\partial g_1},\textrm{\normalfont h}_d((0,\cdot),0)\right\rangle,\vspace{0.1cm}\\
\displaystyle\frac{\partial L_d}{\partial x_1}&=\frac{\partial l_{d+}}{\partial x_1}-\left\langle\frac{\partial l_{d+}}{\partial g_1},\textrm{\normalfont h}_d^0(0,\cdot)\right\rangle, &
\displaystyle\frac{\partial L_d}{\partial g_1}&=\frac{\partial l_{d+}}{\partial g_1}.
\end{alignat*}
In the above expressions, the partial derivatives of $L_d$ and $l_{d+}$ are evaluated at $(x_0,g_0,x_1,g_1')$ and $(x_0,x_1,g_1'-\textrm{\normalfont h}_d((x_0,g_0),x_1))$, respectively. To conclude, we choose $g_0=0$ and $g_1'=g_1+\textrm{\normalfont h}_d^0(x_0,x_1)$.
\end{proof}

Gathering the results of the previous sections, we arrive at the main result of this paper.

\begin{theorem}[$(+)$-discrete Lagrange--Poincar\'e--Dirac equations]
Let $(L_d,X_{d+})$ be a $(+)$-discrete Lagrange--Dirac system and suppose that $L_d$ is $G$-invariant. Let $[\hat D^{d+}]$, $[\hat X_{d+}]$ and $[\hat{\mathcal D}^{+} L_d]$ be the reduced $(+)$-discrete Dirac structure, the $(+)$-discrete reduced vector field and the reduced $(+)$-discrete Dirac differential defined in \eqref{eq:reduceddiracstructure}, \eqref{eq:reducedvectorfield} and \eqref{eq:reduceddiracdifferential}, respectively. Then, they satisfy the \emph{$(+)$-discrete Lagrange--Poincar\'e--Dirac equations}, i.e., for each $0\leq k\leq N-1$, we have
\begin{equation*}
\Big([\hat X_{d+}]\big([q_k,x_{k+1},g_{k+1}]\big),[\hat{\mathcal D}^{+} L_d]\big([q_k,x_k^+,g_k^+]\big)\Big)\in[\hat D^{d+}].
\end{equation*}
\end{theorem}

In order to obtain its local expression, we write the discrete vector field as
\begin{equation*}
[\hat X_{d+}]=\left\{[\hat X_{d+}^k]=(x_k,w_k,\mu_k,q_{k+1},w_{k+1},\mu_{k+1})\mid 0\leq k\leq N-1\right\}.
\end{equation*}
Subsequently, the reduced equations read as
\begin{equation*}
\left([\hat X_{d+}^k],[\hat{\mathcal D}^{+} L_d](x_k,x_k^+,g_k^+)\right)\in[\hat D^{d+}],\qquad 0\leq k\leq N-1.
\end{equation*} 
Making use of Proposition \ref{prop:reduceddiractructurelocal}, Proposition \ref{prop:reduceddiracdifferentiallocal} and Lemma \ref{lemma:partialderivativesld} we arrive at the local expression for the reduced discrete equations of motion,
\begin{equation}\label{eq:geometricequations}
\left\{\begin{array}{l}
\displaystyle \frac{\partial l_{d+}}{\partial x_1}-\left\langle\frac{\partial l_{d+}}{\partial g_1},\textrm{\normalfont h}_d^0(0,\cdot)\right\rangle+\textrm{\normalfont h}_{d,\Sigma}^*\left(\frac{\partial l_{d+}}{\partial g_1}\right)=w_{k+1},\vspace{0.2cm}\\
\displaystyle \frac{\partial l_{d+}}{\partial g_1}=\mu_{k+1},\vspace{0.2cm}\\
\displaystyle -\frac{\partial l_{d+}}{\partial x_0}+\left\langle\frac{\partial l_{d+}}{\partial g_1},\textrm{\normalfont h}_d^0(\cdot,0)\right\rangle=w_k-\textrm{\normalfont h}_{d,\Sigma}^*(\mu_k),\vspace{0.2cm}\\
\displaystyle \left\langle\frac{\partial l_{d+}}{\partial g_1},\textrm{\normalfont h}_d((0,\cdot),0)\right\rangle=\mu_k,\vspace{0.2cm}\\
\displaystyle x_k^+=x_{k+1},\vspace{0.2cm}\\
\displaystyle g_k^++\textrm{\normalfont h}_{d,Q}(x_k,0)=g_{k+1}-\textrm{\normalfont h}_{d,\Sigma}(x_{k+1}).\vspace{0.2cm}\\
\end{array}\right.
\end{equation}
In the above equations, partial derivatives of $l_{d+}$ are evaluated at $(x_k,x_k^+,g_k^+)$.

\subsection{\texorpdfstring{$(-)$}{(--)}-discrete Lagrange--Poincar\'e--Dirac equations}

As in the previous case, we define $\hat L_{d-}: (\Sigma\times G)\times Q\to\mathbb R$ by the condition $\hat L_{d-}\circ\tilde\lambda_d=L_d$, which induces  \emph{$(-)$-discrete reduced Lagrangian}, 
\begin{equation*}
l_{d-}:((\Sigma\times G)\times Q)/G\longrightarrow\mathbb R.
\end{equation*}
Locally, it is given by $l_{d-}:\Sigma\times G\times\Sigma\to\mathbb R$. It is easy to check that
\begin{equation}\label{eq:reducedlagrangianlocalminus}
l_{d-}(x_0,g_0,x_1)=L_d\big(x_0,g_0+\textrm{\normalfont h}_d^0(x_1,x_0),x_1,0\big),\qquad(x_0,g_0,x_1)\in\Sigma\times G\times\Sigma.
\end{equation}

\begin{lemma}\label{lemma:partialderivativesldminus}
Locally, for each $(x_0,x_1,g_1)\in\Sigma\times\Sigma\times G$ we have 
\begin{align*}
\displaystyle \frac{\partial L_d}{\partial q_0} & =\left(\frac{\partial l_{d-}}{\partial x_0}-\left\langle\frac{\partial l_{d-}}{\partial g_0},\textrm{\normalfont h}_d^0(0,\cdot)\right\rangle,\frac{\partial l_{d-}}{\partial g_0}\right),\vspace{0.1cm}\\
\displaystyle \frac{\partial L_d}{\partial x_1} & =\frac{\partial l_{d-}}{\partial x_1}-\left\langle\frac{\partial l_{d-}}{\partial g_0},\textrm{\normalfont h}_d^0(\cdot,0)\right\rangle,\qquad\frac{\partial L_d}{\partial g_1}=-\left\langle\frac{\partial l_{d-}}{\partial g_0},\textrm{\normalfont h}_d((0,\cdot),0)\right\rangle,
\end{align*}
where the partial derivatives of $L_d$ and $l_{d-}$ are evaluated at $\left(x_0,g_0+\textrm{\normalfont h}_d^0(x_1,x_0),x_1,0\right)$ and $(x_0,g_0,x_1)$, respectively.
\end{lemma}

\begin{theorem}[$(-)$-discrete Lagrange--Poincar\'e--Dirac equations]
Let $(L_d,X_{d-})$ be a $(-)$-discrete Lagrange--Dirac system and suppose that $L_d$ is $G$-invariant. Let $[\hat D^{d-}]$, $[\hat X_{d-}]$ and $[\hat{\mathcal D}^{-} L_d]$ be the reduced $(-)$-discrete Dirac structure, the $(-)$-discrete reduced vector field and the reduced $(-)$-discrete Dirac differential defined in \eqref{eq:reduceddiracstructureminus}, \eqref{eq:reducedvectorfieldminus} and \eqref{eq:reduceddiracdifferentialminus}, respectively. Then, they satisfy the \emph{$(-)$-discrete Lagrange--Poincar\'e--Dirac equations}, i.e., for each $0\leq k\leq N-1$, we have
\begin{equation*}
\Big([\hat X_{d-}]\big([x_k,g_k,q_{k+1}]\big),[\hat{\mathcal D}^{-} L_d]\big([x_{k+1}^-,g_{k+1}^-,q_{k+1}]\big)\Big)\in[\hat D^{d-}].
\end{equation*}
\end{theorem}

Locally, the equations are given by
\begin{equation*}
\left([\hat X_{d-}^k],[\hat{\mathcal D}^{-} L_d](x_{k+1}^-,g_{k+1}^-,x_{k+1})\right)\in[\hat D^{d-}],\qquad 0\leq k\leq N-1.
\end{equation*} 
Making use of Proposition \ref{prop:reduceddiractructurelocalminus}, Proposition \ref{prop:reduceddiracdifferentiallocalminus} and Lemma \ref{lemma:partialderivativesldminus} we arrive at the local expression for the reduced discrete equations of motion,
\begin{equation}\label{eq:geometricequationsminus}
\left\{\begin{array}{l}
\displaystyle \frac{\partial l_{d-}}{\partial x_0}-\left\langle\frac{\partial l_{d-}}{\partial g_0},\textrm{\normalfont h}_d^0(0,\cdot)\right\rangle+\textrm{\normalfont h}_{d,\Sigma}^*\left(\frac{\partial l_{d-}}{\partial g_0}\right)=w_k,\vspace{0.2cm}\\
\displaystyle \frac{\partial l_{d-}}{\partial g_0}=\mu_k,\vspace{0.2cm}\\
\displaystyle -\frac{\partial l_{d-}}{\partial x_1}+\left\langle\frac{\partial l_{d-}}{\partial g_0},\textrm{\normalfont h}_d^0(\cdot,0)\right\rangle=w_{k+1}-\textrm{\normalfont h}_{d,\Sigma}^*(\mu_{k+1}),\vspace{0.2cm}\\
\displaystyle \left\langle\frac{\partial l_{d-}}{\partial g_0},\textrm{\normalfont h}_d((0,\cdot),0)\right\rangle=\mu_{k+1},\vspace{0.2cm}\\
\displaystyle x_{k+1}^-=x_k,\vspace{0.2cm}\\
\displaystyle g_{k+1}^-+\textrm{\normalfont h}_{d,Q}(x_{k+1},0)=g_k-\textrm{\normalfont h}_{d,\Sigma}(x_k).
\end{array}\right.
\end{equation}
In the above equations, partial derivatives of $l_{d-}$ are evaluated at $(x_{k+1}^-,g_{k+1}^-,x_{k+1})$.

\section{Reduction of the discrete variational principle}\label{sec:reductionvariationalprinciple}

In this section, we perform reduction of discrete Lagrange--Dirac systems from the variational point of view. As expected, we will recover the discrete Lagrange--Poincar\'e--Dirac equations obtained from the geometric reduction of the discrete Dirac structure. Let $Q$ be a vector space and $G\subset Q$ be a vector subspace acting by addition on $Q$.

\subsection{Trivialization of the \texorpdfstring{$(+)$}{(+)}-discrete Pontryagin bundle}

Let $\omega_d: Q\times Q\to G$ be a discrete connection form. Using the trivializations defined in Section \ref{sec:trivializationsTQ}, we define the following map
\begin{equation*}
\Lambda_{d+}=(\lambda_d,\hat\lambda_d): Q\times Q\times Q\times Q^*\to Q\times(\Sigma\times G)\times Q\times(\Sigma^*\times\mathfrak g^*).
\end{equation*}
Again, the $G$-action on $Q\times Q\times Q\times Q^*$ given by
\begin{equation*}
g\cdot(q_0,q_0^+,q_1,p_1)=(g+q_0,g+q_0^+,g+q_1,p_1),\qquad g\in G,\quad(q_0,q_0^+,q_1,p_1)\in Q\times Q\times Q\times Q^*
\end{equation*}
induces an action on the trivialized space via $\Lambda_{d+}$.

Let $Q=\Sigma\times G$ be a trivialization of $\pi_{Q,\Sigma}$, as in Section \ref{sec:localdiscreteconnection}. For each $q_0=(x_0,g_0),q_0^+=(x_0^+,g_0^+),q_1=(x_1,g_1)\in Q=\Sigma\times G$ and $p_1=(w_1,r_1)\in Q^*=\Sigma^*\times G^*$, the local expression of $\Lambda_{d+}$ is
\begin{equation}\label{eq:Lambdadlocal}
\Lambda_{d+}(q_0,q_0^+,q_1,p_1)=\big(q_0,x_0^+,g_0^+-\textrm{\normalfont h}_d(q_0,x_0^+),q_1,w_1+\textrm{\normalfont h}_{d,\Sigma}^*(r_1),r_1\big).
\end{equation}
Similarly, for each $q_0=(x_0,g_0),q_1=(x_1,g_1)\in Q$, $(x_0^+,g_0^+)\in\Sigma\times G$ and $(w_1,\mu_1)\in\Sigma^*\times\mathfrak g^*$ we have
\begin{equation}\label{eq:invLambdadlocal}
\Lambda_{d+}^{-1}(q_0,x_0^+,g_0^+,q_1,w_1,\mu_1)=\Big(q_0,\big(x_0^+,g_0^++\textrm{\normalfont h}_d(q_0,x_0^+)\big),q_1,\big(w_1-\textrm{\normalfont h}_{d,\Sigma}^*(\mu_1),\mu_1\big)\Big).
\end{equation}
At last, locally we may identify the quotient spaces as
\begin{equation*}
\begin{array}{ccl}
(Q\times Q\times Q\times Q^*)/G & \simeq & \Sigma\times Q\times Q\times Q^*,\\
\left[q_0,q_0^+,q_1,p_1\right] & \mapsto & \left(x_0,-g_0+q_0^+,-g_0+q_1,p_1\right),
\end{array}
\end{equation*}
and
\begin{equation}\label{eq:localquotientQQQQ*}
\begin{array}{ccc}
(Q\times(\Sigma\times G)\times Q\times(\Sigma^*\times\mathfrak g^*))/G & \simeq & \Sigma\times(\Sigma\times G)\times Q\times(\Sigma^*\times\mathfrak g^*),\\
\left[q_0,x_0^+,g_0^+,q_1,w_1,\mu_1\right] & \mapsto & \left(x_0,x_0^+,g_0^+,-g_0+q_1,w_1,\mu_1\right).
\end{array}
\end{equation}

\subsection{Trivialization of the \texorpdfstring{$(-)$}{(-)}-discrete Pontryagin bundle}

In the same vein as in the previous section, we trivialize the $(-)$-discrete Pontryagin bundle. Namely, we define
\begin{equation*}
\Lambda_{d-}=(\tilde\lambda_d,\check\lambda_d): Q\times Q\times Q^*\times Q\to (\Sigma\times G)\times Q\times(\Sigma^*\times\mathfrak g^*)\times Q.
\end{equation*}
Again, the $G$-action on $Q\times Q\times Q^*\times Q$ may be transferred to the trivialized space via $\Lambda_{d-}$.

Let $Q=\Sigma\times G$ be a trivialization of $\pi_{Q,\Sigma}$, as in Section \ref{sec:localdiscreteconnection}. For each $q_0=(x_0,g_0),q_1^-=(x_1^-,g_1^-),q_1=(x_1,g_1)\in Q=\Sigma\times G$ and $p_0=(w_0,r_0)\in Q^*=\Sigma^*\times G^*$, the local expression of $\Lambda_{d+}$ is
\begin{equation}\label{eq:Lambdadlocalminus}
\Lambda_{d-}(q_1^-,q_1,p_0,q_0)=\big(x_1^-,g_1^--\textrm{\normalfont h}_d(q_1,x_1^-),q_1,w_0+\textrm{\normalfont h}_{d,\Sigma}^*(r_0),r_0,q_0\big).
\end{equation}
Similarly, for each $q_0=(x_0,g_0),q_1=(x_1,g_1)\in Q$, $(x_1^-,g_1^-)\in\Sigma\times G$ and $(w_0,\mu_0)\in\Sigma^*\times\mathfrak g^*$ we have
\begin{equation}\label{eq:invLambdadlocalminus}
\Lambda_{d-}^{-1}(x_1^-,g_1^-,q_1,w_0,\mu_0,q_0)=\Big((x_1^-,g_1^-+\textrm{\normalfont h}_d(q_1,x_1^-)),q_1,(w_0-\textrm{\normalfont h}_{d,\Sigma}^*(\mu_0),\mu_0),q_0\Big).
\end{equation}
At last, locally we may identify the quotient spaces as
\begin{equation*}
(Q\times Q\times Q^*\times Q)/G\simeq Q\times\Sigma\times Q^*\times Q,
\end{equation*}
and
\begin{equation}\label{eq:localquotientQQQQ*minus}
((\Sigma\times G)\times Q\times(\Sigma^*\times\mathfrak g^*)\times Q)/G\simeq (\Sigma\times G)\times\Sigma\times(\Sigma^*\times\mathfrak g^*)\times Q.
\end{equation}

\subsection{\texorpdfstring{$(+)$}{(+)}-discrete reduced variational principle}

Let $L_d: Q\times Q\to\mathbb R$ be a (possibly degenerate) $G$-invariant discrete Lagrangian and consider the corresponding $(+)$-discrete reduced Lagrangian $l_{d+}:(Q\times(\Sigma\times G))/G\to\mathbb R$. The \emph{$(+)$-discrete generalized energy} is the map
\begin{equation}\label{eq:Ed}
E_{d+}: Q\times Q\times Q\times Q^*\to\mathbb R,\qquad (q_0,q_0^+,q_1,p_1)\mapsto L_d(q_0,q_0^+)+\langle p_1,q_1-q_0^+\rangle.
\end{equation}
Since $L_d$ is $G$-invariant, so is $E_{d+}$. Analogous to the discrete Lagrangian, we consider the trivialized energy, i.e., $\hat E_{d+}: Q\times(\Sigma\times G)\times Q\times(\Sigma^*\times\mathfrak g^*)\to\mathbb R$, which is defined by the condition $\hat E_{d+}\circ\Lambda_{d+}=E_{d+}$. Again, $\hat E_{d+}$ is $G$-invariant, what enables us define the \emph{reduced $(+)$-discrete generalized energy},
\begin{equation*}
e_{d+}:(Q\times(\Sigma\times G)\times Q\times(\Sigma^*\times\mathfrak g^*))/G\longrightarrow\mathbb R.
\end{equation*}

\begin{lemma}\label{lemma:reducedactionlocal}
Locally, the reduced $(+)$-discrete generalized energy is a map $e_{d+}:\Sigma\times(\Sigma\times G)\times Q\times(\Sigma^*\times\mathfrak g^*)\to\mathbb R$. Furthermore, for each $(x_0,x_0^+,g_0^+,q_1,w_1,\mu_1)\in\Sigma\times(\Sigma\times G)\times Q\times(\Sigma^*\times\mathfrak g^*)$ we have
\begin{align*}
e_{d+} (x_0,x_0^+,g_0^+,q_1,w_1,\mu_1)&=
l_{d+}\left(x_0,x_0^+,g_0^+\right)+\big\langle w_1-\textrm{\normalfont h}_{d,\Sigma}^*(\mu_1),x_1-x_0^+\big\rangle\\
&\qquad\qquad+\big\langle\mu_1,g_1-g_0^+-\textrm{\normalfont h}_d((x_0,0),x_0^+)\big\rangle.
\end{align*}
\end{lemma}

\begin{proof}
It is a straightforward computation, making use of the local expressions of the maps and actions that we have computed previously,
\begin{equation*}
\begin{array}{c}
(x_0,x_0^+,g_0^+,q_1,w_1,\mu_1)\\
\downmapsto\text{\scriptsize\eqref{eq:localquotientQQQQ*}}\\
\big((x_0,0),x_0^+,g_0^+,q_1,w_1,\mu_1\big)\\
\downmapsto\text{\scriptsize\eqref{eq:invLambdadlocal}}\\
\Big((x_0,0),\big(x_0^+,g_0^++\textrm{\normalfont h}_d((x_0,0),x_0^+)\big),q_1,\big(w_1-\textrm{\normalfont h}_{d,\Sigma}^*(\mu_1),\mu_1\big)\Big)\\
\downmapsto\text{\scriptsize\eqref{eq:Ed}}\\
L_d\big(x_0,0,x_0^+,g_0^++\textrm{\normalfont h}_d((x_0,0),x_0^+)\big)+\Big\langle\big(w_1-\textrm{\normalfont h}_{d,\Sigma}^*(\mu_1),\mu_1\big),q_1-\big(x_0^+,g_0^++\textrm{\normalfont h}_d((x_0,0),x_0^+)\big)\Big\rangle\\
\downmapsto\text{\scriptsize\eqref{eq:reducedlagrangianlocal}}\\
l_{d+}\left(x_0,x_0^+,g_0^+\right)+\big\langle w_1-\textrm{\normalfont h}_{d,\Sigma}^*(\mu_1),x_1-x_0^+\big\rangle+\big\langle\mu_1,g_1-g_0^+-\textrm{\normalfont h}_d((x_0,0),x_0^+)\big\rangle
\end{array}
\end{equation*}
\end{proof}

The following result relates the variational principles in both the original and the reduced spaces.

\begin{theorem}[Reduced variational principle]
Let $L_d: Q\times Q\to\mathbb R$ be a $G$-invariant discrete Lagrangian and $\left\{(q_k,q_k^+,p_{k+1})\in Q\times Q\times Q^*\mid 0\leq k\leq N\right\}$ be a trajectory on the $(+)$-discrete Pontryagin bundle. Consider the reduced trajectory, i.e.,
\begin{equation*}
\left\{\big[\hat q_k,\hat x_k^+,\hat g_k^+,\hat q_{k+1},\hat w_{k+1},\hat \mu_{k+1}\big]\in\big(Q\times(\Sigma\times G)\times Q\times(\Sigma^*\times\mathfrak g^*)\big)/G\mid 0\leq k\leq N-1\right\},    
\end{equation*}
where $(\hat q_k,\hat x_k^+,\hat g_k^+,\hat q_{k+1},\hat w_{k+1},\hat \mu_{k+1})=\Lambda_{d+}(q_k,q_k^+,q_{k+1},p_{k+1})$ for $0\leq k\leq N-1$. Then the $(+)$-discrete Lagrange--Pontryagin principle \eqref{eq:lagrangepontryaginprinciple} is satisfied if and only if the \emph{reduced $(+)$-discrete Lagrange--Pontryagin principle} is satisfied, i.e.,
\begin{equation}\label{eq:reducedlagrangepontryaginprinciple}
\delta\sum_{k=0}^{N-1}e_{d+} \left(\big[\hat q_k,\hat x_k^+,\hat g_k^+,\hat q_{k+1},\hat w_{k+1},\hat \mu_{k+1}\big]\right)=0,
\end{equation}
for free variations $\left\{\left(\delta\hat q_k,\delta \hat x_k^+,\delta \hat g_k^+,\delta \hat w_{k+1},\delta\hat\mu_{k+1}\right)\in Q\times(\Sigma\times G)\times(\Sigma^*\times\mathfrak g^*)\mid 0\leq k\leq N\right\}$ with fixed endpoints, i.e., $\delta\hat q_0=\delta\hat q_N=0$.
\end{theorem}

\begin{proof}
By construction, we have
\begin{align*}
\displaystyle \mathbb S_{L_d}^+\left[(q_k,q_k^+,p_{k+1})_{k=0}^N\right] & = \displaystyle \sum_{k=0}^{N-1}(\hat E_{d+}\circ\Lambda_{d+})(q_k,q_k^+,q_{k+1},p_{k+1})\vspace{0.1cm}\\
& = \displaystyle \sum_{k=0}^{N-1}\hat E_{d+}\left(\hat q_k,\hat x_k^+,\hat g_k^+,\hat q_{k+1},\hat w_{k+1},\hat\mu_{k+1}\right)\vspace{0.1cm}\\
& = \displaystyle \sum_{k=0}^{N-1}e_{d+} \left(\big[\hat q_k,\hat x_k^+,\hat g_k^+,\hat q_{k+1},\hat w_{k+1},\hat \mu_{k+1}\big]\right).
\end{align*}
To arrive at our conclusion, note that the free variations of the Lagrange--Pontryagin principle yield free variations on the trivialized space (with fixed endpoints), since $\Lambda_{d+}$ is a linear isomorphism.
\end{proof}

Lastly, we show that the reduced variational equations agree with the reduced geometric equations obtained in the previous section from the reduced discrete Dirac structure.

\begin{proposition}
The variational equations obtained from the reduced $(+)$-discrete Lagrange--Pontryagin principle are the $(+)$-discrete Lagrange--Poincar\'e--Dirac equations.
\end{proposition}

\begin{proof}
Since the equations are local, we may work in a trivialization $Q\simeq\Sigma\times G$ of $\pi_{Q,\Sigma}$. Observe that locally, \eqref{eq:reducedlagrangepontryaginprinciple} reads as
\begin{equation*}
\delta\sum_{k=0}^{N-1}e_{d+} (x_k,x_k^+,g_k^+,q_{k+1},w_{k+1},\mu_{k+1})=0,
\end{equation*}
for free variations $\left\{(\delta x_k,\delta g_k,\delta x_k^+,\delta g_k^+,\delta w_{k+1},\delta\mu_{k+1})\in Q\times(\Sigma\times G)\times(\Sigma^*\times\mathfrak g^*)\mid 0\leq k\leq N\right\}$ with fixed endpoints, i.e., $\delta x_0=\delta x_N=0$ and $\delta g_0=\delta g_N=0$. Making use of Lemma \ref{lemma:reducedactionlocal}
and taking variations $\delta x_k$, $1\leq k\leq N-1$, with fixed endpoints, we get
\begin{equation*}
\frac{\partial l_{d+}}{\partial x_0}+w_k-\textrm{\normalfont h}_{d,\Sigma}^*(\mu_k)-\left\langle\mu_{k+1},\textrm{\normalfont h}_d^0(\cdot,0)\right\rangle=0.
\end{equation*}
Analogously, taking variations $\delta g_k$, $1\leq k\leq N-1$, with fixed endpoints, yield
\begin{equation*}
\mu_k-\left\langle\mu_{k+1},\textrm{\normalfont h}_d((0,\cdot),0)\right\rangle=0.
\end{equation*}
Now, we consider variations $\delta x_k^+$, $0\leq k\leq N-1$,
\begin{equation*}
\frac{\partial l_{d+}}{\partial x_1}-w_{k+1}+\textrm{\normalfont h}_{d,\Sigma}^*(\mu_{k+1})-\left\langle\mu_{k+1},\textrm{\normalfont h}_d^0(0,\cdot)\right\rangle=0.
\end{equation*}
Similarly, for variations $\delta g_k^+$, $0\leq k\leq N-1$,
\begin{equation*}
\frac{\partial l_{d+}}{\partial g_1}-\mu_{k+1}=0.
\end{equation*}
Next, for variations $\delta w_{k+1}$, $0\leq k\leq N-1$,
\begin{equation*}
x_{k+1}-x_k^+=0.
\end{equation*}
In the end, for variations $\delta\mu_{k+1}$, $0\leq k\leq N-1$,
\begin{equation*}
\textrm{\normalfont h}_{d,\Sigma}(x_{k+1}-x_k^+)+g_{k+1}-g_k^+- \textrm{\normalfont h}_d^0(x_k,x_k^+)=0.
\end{equation*}
In the above equations, partial derivatives of $l_{d+}$ are evaluated at $(x_k,x_k^+,g_k^+)$. By gathering all the equations and rearranging terms, it is easy to check that these equations are exactly \eqref{eq:geometricequations}.
\end{proof}

\subsection{\texorpdfstring{$(-)$}{(-)}-discrete reduced variational principle}

Last of all, we carry out the same procedure, but using the $(-)$-discrete reduced Lagrangian $l_{d-}:((\Sigma\times G)\times Q)/G\to\mathbb R$. The \emph{$(-)$-discrete generalized energy} is the map
\begin{equation}\label{eq:Edminus}
E_{d-}: Q\times Q\times Q^*\times Q\to\mathbb R,\qquad (q_1^-,q_1,p_0,q_0)\mapsto L_d(q_1^-,q_1)+\langle p_0,q_0-q_1^-\rangle.
\end{equation}
The trivialized energy, $\hat E_{d-}\circ\Lambda_{d-}=E_{d-}$, gives rise to the \emph{reduced $(-)$-discrete generalized energy},
\begin{equation*}
e_{d-}:((\Sigma\times G)\times Q\times(\Sigma^*\times\mathfrak g^*)\times Q)/G\longrightarrow\mathbb R.
\end{equation*}

\begin{lemma}\label{lemma:reducedactionlocalminus}
Locally, the reduced $(-)$-discrete generalized energy is a map $e_{d-}:(\Sigma\times G)\times\Sigma\times(\Sigma^*\times\mathfrak g^*)\times Q\to\mathbb R$. Furthermore, for each $(x_1^-,g_1^-,x_1,w_0,\mu_0,q_0)\in(\Sigma\times G)\times\Sigma\times(\Sigma^*\times\mathfrak g^*)\times Q$ we have
\begin{align*}
e_{d-} (x_1^-,g_1^-,x_1,w_0,\mu_0,q_0) & = l_{d-}(x_1^-,g_1^-,x_1)+\big\langle w_0-\textrm{\normalfont h}_{d,\Sigma}^*(\mu_0),x_0-x_1^-\big\rangle\\
& \qquad\qquad+\big\langle\mu_0,g_0-g_1^--\textrm{\normalfont h}_d((x_1,0),x_1^-)\big\rangle.
\end{align*}
\end{lemma}

\begin{theorem}[Reduced variational principle]
Let $L_d: Q\times Q\to\mathbb R$ be a $G$-invariant discrete Lagrangian and $\left\{(q_{k+1}^-,p_k,q_{k+1})\in Q\times Q^*\times Q\mid 0\leq k\leq N\right\}$ be a trajectory on the $(-)$-discrete Pontryagin bundle. Consider the reduced trajectory, i.e.,
\begin{equation*}
\left\{\big[\hat x_{k+1}^-,\hat g_{k+1}^-,\hat q_{k+1},\hat w_k,\hat \mu_k,\hat q_k\big]\in\big((\Sigma\times G)\times Q\times(\Sigma^*\times\mathfrak g^*)\times Q\big)/G\mid 0\leq k\leq N-1\right\},    
\end{equation*}
where $(\hat x_{k+1}^-,\hat g_{k+1}^-,\hat q_{k+1},\hat w_k,\hat \mu_k,\hat q_k)=\Lambda_{d-}(q_{k+1}^-,q_{k+1},p_k,q_k)$ for $0\leq k\leq N-1$. Then the $(-)$-discrete Lagrange--Pontryagin principle \eqref{eq:lagrangepontryaginprincipleminus} is satisfied if and only if the \emph{reduced $(-)$-discrete Lagrange--Pontryagin principle} is satisfied, i.e.,
\begin{equation}\label{eq:reducedlagrangepontryaginprincipleminus}
\delta\sum_{k=0}^{N-1}e_{d-} \left([\hat x_{k+1}^-,\hat g_{k+1}^-,\hat q_{k+1},\hat w_k,\hat \mu_k,\hat q_k\big]\right)=0,
\end{equation}
for free variations $\left\{\left(\delta\hat x_{k+1}^-,\delta\hat g_{k+1}^-,\delta\hat q_{k+1},\delta\hat w_k,\delta\hat \mu_k,\delta\hat q_k\right)\in (\Sigma\times G)\times Q\times(\Sigma^*\times\mathfrak g^*)\times Q\mid 0\leq k\leq N\right\}$ with fixed endpoints, i.e., $\delta\hat q_0=\delta\hat q_N=0$.
\end{theorem}

\begin{proposition}
The variational equations obtained from the reduced $(-)$-discrete Lagrange--Pontryagin principle are the $(-)$-discrete Lagrange--Poincar\'e--Dirac equations.
\end{proposition}

\section{Nonlinear theory}\label{sec:nonlinear}

The previous reduction theory has been developed for the linear setting, i.e., when $Q$ is a vector space and $G\subset Q$ is a vector subspace acting by addition on $Q$. Nevertheless, it can be applied when $Q$ is an arbitrary smooth manifold and $G$ is an abelian Lie group acting freely and properly on $Q$. In order to see this, we use retractions and retraction compatible charts (see, for example, \cite{absil2008,leok2011}).

\begin{definition}
A \emph{retraction} of a smooth manifold $M$ is a smooth map $\mathcal R: TM\to M$ such that for each $m\in M$ we have $\mathcal R_m(0_m)=m$ and $\left(d\mathcal R_m\right)_{0_m}=\textrm{\normalfont id}_{T_m M}$, where $\mathcal R_m=\mathcal R\vert_{T_m M}$ and we make the identification $T_{0_m}(T_m M)\simeq T_m M$.
\end{definition}

Observe that the second condition ensures that $\mathcal R_m: T_m M\to M$ is invertible around $0_m$.

\begin{definition}
Let $M$ be an $n$-dimensional smooth manifold and $\mathcal R: TM\to M$ be a retraction of $M$. A coordinate chart $(U,\phi)$ on  $M$ is said to be \emph{compatible at $m\in M$} with $\mathcal R$ if $\phi(m)=0$ and $\mathcal R(v_m)=\phi^{-1}\left((d\phi)_m(v_m)\right)$ for each $v_m\in T_m M$, where we identify $T_m\mathbb R^n\simeq\mathbb R^n$ using the standard coordinates in $\mathbb R^n$.
\end{definition}

For the particular case of a Lie group $G$ and $g\in G$, it was shown in \cite[\S9]{leok2011} that the \emph{canonical coordinates of the first kind} at $g\in G$ (cf. \cite{varadarajan2013,MaPeSh1999}) are a coordinate chart compatible with the retraction defined by $\mathcal R_g^G=L_g\circ\exp\circ(dL_{g^{-1}})_g$, where $L_g: G\to G$ denotes the left multiplication by $g$.

\begin{proposition}\label{prop:mainfeature}
Let $\mathcal R: TM\to M$ be a retraction of an $n$-dimensional smooth manifold $M$ and $(U,\phi)$ be a compatible coordinate chart on $M$ at $m\in M$. Then for each\footnote{
We will assume that the inverse of $\mathcal R_m$ is defined on the whole of $U$ by choosing a smaller coordinate domain if necessary.} $r\in U$ and $p_m\in T_m^* M$ we have
\begin{equation*}
\left\langle p_m,\mathcal R_m^{-1}(r)\right\rangle=\sum_{i=1}^n p_i r^i,
\end{equation*}
where $\mathcal R_m^{-1}(r)\simeq r^i\partial_i$ and $p_m\simeq p_i dq^i$ in this chart. 
\end{proposition}

In other words, the previous result says that the dual pairing of $T_m M$ and $T_m^* M$ reduces to the usual Euclidean inner product on $\mathbb R^n$ when using retraction compatible charts.

\subsection{Retractions and abelian Lie group actions}

In what follows, let $Q$ be a smooth manifold and $G$ be a connected, abelian Lie group acting freely and properly (on the left) on $Q$, thus yielding a principal bundle $\pi_{Q,\Sigma}: Q\to\Sigma=Q/G$. Consider retractions $\mathcal R^\Sigma: T\Sigma\to\Sigma$ and $\mathcal R^G: TG\to G$  of $\Sigma$ and $G$, respectively, and a trivializing set $U\subset\Sigma$ of $\pi_{Q,\Sigma}$. For simplicity, we write $U=\Sigma$, so we have an identification $Q\simeq\Sigma\times G$. Under this identification, a straightforward check shows that the map
\begin{equation*}
\mathcal R=\left(\mathcal R^\Sigma,\mathcal R^G\right): TQ\to Q
\end{equation*}
is a retraction of $Q$. We may use it to define (at least locally) a discrete Lagrange--Pontryagin action. Namely, given a (possibly degenerate) discrete Lagrangian $L_d: Q\times Q\to\mathbb R$, the \emph{$(+)$-discrete Lagrange--Pontryagin action} is defined as
\begin{equation*}
\mathbb S_{L_d}^+\left[\left(q_k,q_k^+,p_{k+1}\right)_{k=0}^N\right]=\sum_{k=0}^{N-1}\left(L_d\left(q_k,q_k^+\right)+\left\langle p_{k+1},\mathcal R_{q_{k+1}}^{-1}(q_{k+1})-\mathcal R_{q_{k+1}}^{-1}(q_k^+)\right\rangle\right),
\end{equation*}
where $q_k,q_k^+\in Q$ and $p_{k+1}\in T_{q_{k+1}}^*Q$. The $(+)$-discrete Lagrange--Pontryagin principle is obtained by enforcing free variations vanishing at the endpoints. In order for the previous expression to be well-defined, $q_k^+$ must be in the neighborhood of $q_{k+1}$ where $\mathcal R_{q_{k+1}}$ is invertible for each $0\leq k\leq N-1$. This can always be achieved for bounded initial conditions by reducing the size of the time step. Furthermore, Proposition \ref{prop:mainfeature} ensures that using retraction compatible coordinate charts, this discrete action reduces to the one considered in the linear case \eqref{eq:lagrangepontryaginaction}.

In the same vein, the \emph{$(-)$-discrete Lagrange--Pontryagin action} is defined as
\begin{equation*}
\mathbb S_{L_d}^-\left[(q_{k+1}^-,p_k,q_{k+1})_{k=0}^N\right]=\sum_{k=0}^{N-1}\left(L_d(q_{k+1}^-,q_{k+1})+\left\langle p_k,\mathcal R_{p_k}^{-1}(q_k)-\mathcal R_{p_k}^{-1}(q_{k+1}^-)\right\rangle\right),
\end{equation*}
where $q_{k+1}^-,q_{k+1}\in Q$ and $p_k\in Q^*$. Once again, Proposition \ref{prop:mainfeature} ensures that, in a retraction compatible chart, the discrete action reads as in the linear case \eqref{eq:lagrangepontryaginactionminus}.

On the other hand, recall that the $G$ action on the trivialized principal bundle is given by the left multiplication, i.e., $g\cdot q_0=(x,gg_0)$ for each $g\in G$ and $q_0=(x_0,g_0)\in Q\simeq \Sigma\times G$. Furthermore, suppose that the retraction on $G$ is given by $\mathcal R_{g_0}^G=L_{g_0}\circ\exp\circ(dL_{g_0^{-1}})_{g_0}$, as mentioned when describing the canonical coordinates of the first kind. Since $G$ is abelian and connected, the exponential map is surjective and, hence, the inverse of $\mathcal R_{g_0}^G$ exists locally around every element of the group. As a result, when $q_0=(x_0,g_0)$ is fixed, an action of $T_{g_0}G$ on $T_{q_0}Q=T_{x_0}\Sigma\oplus T_{g_0}G$ may be built as follows
\begin{equation*}
u_{g_0}\cdot(v_{x_0},v_{g_0})=\left(v_{x_0},\left(\mathcal R_{g_0}^G\right)^{-1}\left(\mathcal R_{g_0}^G(u_{g_0})\mathcal R_{g_0}^G(v_{g_0})\right)\right),\qquad v_{x_0}\in T_{x_0}\Sigma,\quad u_{g_0},v_{g_0}\in T_{g_0}G.
\end{equation*}
When $g_0=0$, this action has a very simple expression,
\begin{equation*}
\xi\cdot(v_{x_0},\xi_0)=\left(v_{x_0},\log(\exp(\xi)\exp(\xi_0))\right)=(v_{x_0},\xi+\xi_0),\qquad v_{x_0}\in T_{x_0}\Sigma,\quad\xi,\xi_0\in\mathfrak g,
\end{equation*}
where we have used that $G$ is abelian. This is the case considered in the previous sections, a vector subspace acting by addition. Note that if $g_0\neq 0$, we obtain the same result by using the left translation.

In summary, the linear theory developed above is the coordinate representation of a general system with abelian group of symmetries when retraction compatible charts are used. This enables to use the linear theory at least locally in a coordinate chart. Furthermore, if $\Sigma$ and $G$ both admit retraction compatible atlases, then a retraction compatible atlas for $Q$ can be built such that each coordinate domain is a trivializing set for $\pi_{Q,\Sigma}$. Hence, we can make computations on the whole $Q$ by starting from a specific compatible chart and changing to another one whenever it is necessary. Since the variational principle, as well as the group action, have the same expressions in each compatible chart, the local preservation of geometric properties extends to the global setting. This is an important feature of retraction compatible atlases, since computation on local charts might otherwise lead to dynamics that is not globally well-defined.

\section{Examples}\label{sec:examples}

We present two examples to illustrate the reduction theory developed above. In both of them, we employ the $(+)$-discrete equations.

\subsection{Charged particle in a magnetic field}

We analyze the dynamics of a charged particle moving in a magnetic field, as presented in \cite{marsdenjurgen1993} for the continuous case. To account for gauge symmetry, we consider the \emph{Kaluza-Klein configuration space},
\begin{equation*}
Q_K=\mathbb R^3\times\mathbb S^1,
\end{equation*}
with coordinates $(q,\theta)$, where $q=(q^1,q^2,q^3)\in\mathbb R^3$ and $\theta\simeq e^{i\theta}\in\mathbb S^1$. The \emph{Kaluza-Klein Lagrangian} is defined as
\begin{equation*}
L_K(q,\dot q,\theta,\dot\theta)=\frac{1}{2}m\llangle\dot q,\dot q\rrangle+\frac{1}{2}\left(\langle A(q),\dot q\rangle+\dot\theta\right)^2,
\end{equation*}
where $m\in\mathbb R^+$ is the mass of the particle, $A\in\Omega^1(\mathbb R^3)$ is the magnetic potential, and $\llangle\cdot,\cdot\rrangle$ denotes the Euclidean inner product in $\mathbb R^3$. The conjugate momenta are given by
\begin{equation*}
p=\frac{\partial L_K}{\partial\dot q}=m\llangle\dot q,\cdot\rrangle+\left(\langle A(q),\dot q\rangle+\dot\theta\right)A(q),\qquad p_\theta=\frac{\partial L_K}{\partial\dot\theta}=\langle A(q),\dot q\rangle+\dot\theta.
\end{equation*}

It is clear that this Lagrangian is invariant under the tangent lifted action of $\mathbb S^1$ on $Q_K$ given by $\theta'\cdot(q,\theta)=(q,\theta'+\theta)$ for each $(q,\theta)\in Q_K$ and $\theta'\in\mathbb S^1$. Note that using these coordinates, the action is linear. The corresponding quotient is $\Sigma=Q_K/\mathbb S^1\simeq\mathbb R^3$. Choosing the principal connection $\omega=A+d\theta\in\Omega^1(Q_K)$ on $Q_K\to\mathbb R^3$, the corresponding reduced equations are the \emph{Lorentz force law} together with the conservation of the momentum corresponding to $\theta$, i.e.,
\begin{equation}\label{eq:lorentzlaw}
m\frac{d\dot q}{dt}=\frac{e}{c}\dot q\times\textbf B,\qquad\dot p_\theta=0,
\end{equation}
where $\textbf B\in\mathfrak X(\mathbb R^3)$ is the magnetic field corresponding to\footnote{
Recall that we can associate to each 2-form $\alpha=\alpha^i\iota_{\partial_i}(d^3q)\in\Omega^2(\mathbb R^3)$ a vector field $\boldsymbol\alpha=\alpha^i\partial_i\in\mathfrak X(\mathbb R^3)$, which is known as proxy field, where $\iota_U$ denotes the left interior product by $U\in\mathfrak X(\mathbb R^3)$.} $B=dA\in\Omega^2(\mathbb R^3)$, and $e=c\,p_\theta$ is the electric charge of the particle.

\subsubsection{Discrete equations}

Given $h> 0$, consider the following discrete Lagrangian
\begin{equation*}
L_d(q_0,\theta_0,q_0^+,\theta_0^+;h)=hL_K\left(q_0,\frac{q_0^+-q_0}{h},\theta_0,\frac{\theta_0^+-\theta_0}{h}\right).
\end{equation*}
It is easy to check from Definition \ref{def:discreteconnection} that the following is a discrete principal connection,
\begin{equation*}
\omega_d((q_0,\theta_0),(q_0^+,\theta_0^+))=\theta_0^+-\theta_0,\qquad(q_0,\theta_0),(q_0^+,\theta_0^+)\in Q_K.
\end{equation*}
Subsequently, the local map $\textrm{\normalfont h}_d: Q_K\times\mathbb R^3\to\mathbb S^1$ is given by $\textrm{\normalfont h}_d((q_0,\theta_0),q_0^+)=\theta_0$ and $\textrm{\normalfont h}_{d,\Sigma}\equiv0$. The reduced discrete Lagrangian \eqref{eq:reducedlagrangianlocal} is
\begin{align*}
l_d(q_0,q_0^+,\theta_0^+;h) & = L_d(q_0,0,q_0^+,\theta_0^++\textrm{\normalfont h}_d^0(q_0,q_0^+))\vspace{0.1cm}\\
& = \displaystyle \frac{m}{2h}\left\llangle q_0^+-q_0,q_0^+-q_0\right\rrangle+\frac{1}{2h}\left(\langle A(q_0),q_0^+-q_0\rangle+\theta_0^+\right)^2.
\end{align*}
Consider an interval $[0,T]\subset\mathbb R$ and divide it into $N=T/h$ subintervals $[t_k,t_{k+1}]$, with $t_k=kh$, $0\leq k\leq N-1$. As usual, we let $q_k=q(t_k)$, $0\leq k\leq N$, and analogously for the other quantities. Similarly, we express the discrete vector field as $[\hat X_d^k]=(q_k,w_k,\mu_k,q_{k+1},\theta_{k+1},w_{k+1},\mu_{k+1})$. By computing the partial derivatives of the reduced discrete Lagrangian, we obtain the reduced discrete equations
\begin{equation}\label{eq:reduceddiscreteequationsem}
\left\{\begin{array}{l}
\displaystyle \frac{m}{h}\llangle q_{k+1}-q_k,\cdot\rrangle+\mu_{k+1}A(q_k)=w_{k+1},\vspace{0.2cm}\\
\displaystyle \frac{1}{h}(\langle A(q_k),q_{k+1}-q_k\rangle+\theta_{k+1})=\mu_{k+1},\vspace{0.2cm}\\
\displaystyle \frac{m}{h}\llangle q_{k+1}-q_k,\cdot\rrangle-\mu_{k+1}\left(\langle A,q_{k+1}-q_k\rangle-A(q_k)\right)=w_k,\vspace{0.2cm}\\
\displaystyle \mu_{k+1}=\mu_k,
\end{array}\right.,\qquad 0\leq k\leq N-1.
\end{equation}

\subsubsection{Numerical computations}

\begin{table}[b]
    \begin{tabular}{c|c}
    $N$ & $||q(T)-q_N||$\\
    \hline
    10 & 1.6781\\
    50 & 0.25971\\
    100 & 0.06626\\
    200 & 0.01664
    \end{tabular}
    \caption{Number of steps vs. error at the last step}
    \label{table:errorcomparison}
\end{table}

In order to implement the above equations, we need to make a particular choice of the particle mass and charge, time interval, magnetic field and initial conditions. We express all of the magnitudes in natural units, i.e., $c=1$. We choose $m=1$, $e=1$ and $T=20$. For each $q=(q^1,q^2,q^3)\in\mathbb R^3$, we write $q=q^i\partial_i\in T_q\mathbb R^3\simeq\mathbb R^3$ when regarded as a vector. Hence, $\llangle q,\cdot\rrangle=q^i dq^i$. On the other hand, we suppose that the magnetic field is constant, i.e., $\textrm{\normalfont B}(q)=B_0\partial_z$ for some fixed $B_0\in\mathbb R$. The corresponding magnetic potential is \begin{equation*}
A(q)=\frac{B_0}{2}\left(-q^2 dq^1+q^1 dq^2\right),\qquad q=\left(q^1,q^2,q^3\right)\in\mathbb R^3.
\end{equation*}

We suppose that the initial position is the origin, $q_0=0$, and the initial velocity is $\dot q_0=\partial_1+\partial_3$. Likewise, we choose $\theta_0=0$. Since $A(q_0)=A(0)=0$, the corresponding initial momenta are 
\begin{equation*}
p_0=\llangle\dot q_0,\cdot\rrangle+\left(\langle A(q_0),\dot q_0\rangle+\dot\theta\right)A(q_0)=\llangle\dot q_0,\cdot\rrangle=dq^1+dq^3,
\end{equation*}
and $(p_\theta)_0=e=1$. Likewise, we have $\dot\theta_0=(p_\theta)_0-\langle A(q_0),\dot q_0\rangle=1$. The initial position and momenta may be trivialized using \eqref{eq:hatlambdadlocal},
\begin{equation*}
\big((q_0,\theta_0),w_0,\mu_0\big)=\hat\lambda_d\big((q_0,\theta_0),(p_0,(p_\theta)_0)\big)=\big((0,0),(1,0,1),1\big).
\end{equation*}

On the other hand, it is easy to check that the exact solution of \eqref{eq:lorentzlaw} for $B_0=1$ is 
\begin{equation*}
q(t)=\left(\sin\left(t\right),\cos\left(t\right)-1,t\right),\qquad 0\leq t\leq T.
\end{equation*}
After working out the approximate solution for different values of $N$ and comparing to the exact solution, we can see that the error decreases with the number of steps at a second-order rate of convergence, as shown in Table \ref{table:errorcomparison}. In addition, Figure \ref{fig:chargedparticle} compares both the exact and numerical trajectories for $N=100$, i.e., $h=0.2$.

\begin{figure}[t]
    \centering
    \includegraphics{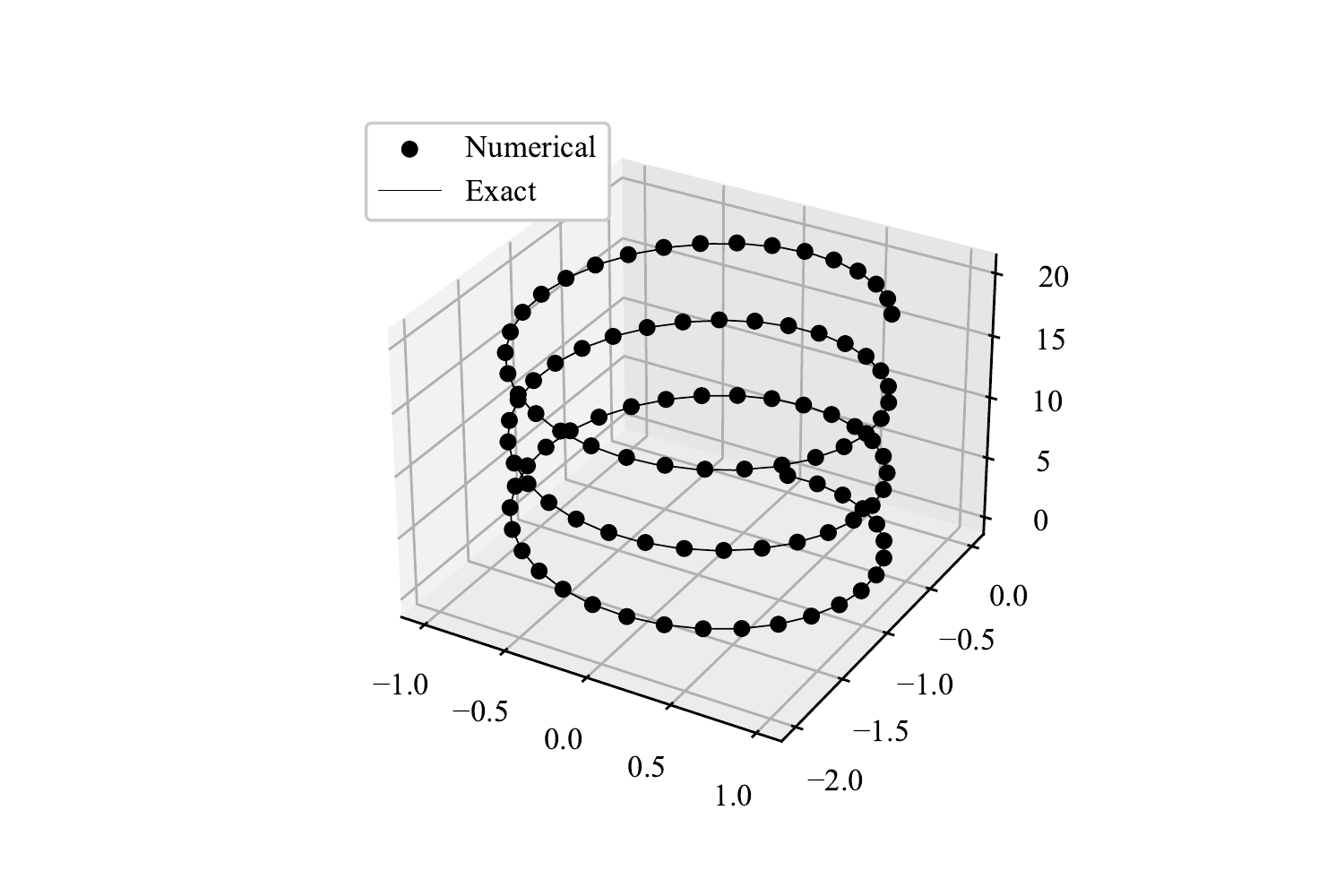}
    \caption{Exact and approximate trajectories of a charged particle moving in a constant vector field. The approximate trajectory is computed using time step $h=0.2$.}
    \label{fig:chargedparticle}
\end{figure}

\subsection{Double spherical pendulum}

In this example, we investigate the double spherical pendulum. We assume that there is not friction and that the system is under a uniform gravitational field. The dynamics of the double spherical pendulum has been investigated in \cite{marsden1992,marsden1993} and variational integrators from different perspectives have been proposed in \cite{jalnapurkar2005,leok2009}.

For $i=1,2$, we denote by $m_i\in\mathbb R^+$, $l_i\in\mathbb R^+$ and $r_i\in\mathbb R^3$ the particle mass, the link length  and the position of the $i$-th pendulum, respectively. Furthermore, we use standard coordinates $r_i=(r_i^1,r_i^2,r_i^3)\in\mathbb R^3$ with coordinate origin at the fixed pivot, and we suppose that the gravitational acceleration is given by $\boldsymbol{g}=(0,0,-g)$ for some fixed $g\in\mathbb R^+$. This way, the Lagrangian is
\begin{equation*}
L(r_1,r_2,\dot r_1,\dot r_2)=\frac{1}{2}m_1\llangle\dot r_1,\dot r_1\rrangle+\frac{1}{2}m_2\llangle\dot r_2,\dot r_2\rrangle-m_1\,g\,r_1^3-m_2\,g\,r_2^3,
\end{equation*}
where $\llangle\cdot,\cdot\rrangle$ denotes the Euclidean inner product in $\mathbb R^3$. In addition, the links connecting the particles yield the following constraints,
\begin{equation*}
\llangle r_1,r_1\rrangle=l_1^2,\qquad \llangle r_2-r_1,r_2-r_1\rrangle=l_2^2.
\end{equation*}

\begin{figure}[ht]
    \centering
    \includegraphics[width=0.5\textwidth]{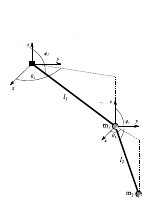}
    \caption{Spherical coordinates for the double spherical pendulum.}
    \label{fig:doublesphericalpendulum}
\end{figure}

To avoid the constraints we use spherical coordinates for each particle, $r_i=(\rho_i,\theta_i,\varphi_i)$, $i=1,2$, where the origin of the first sphere is at the pivot and the origin of the second one is at the first particle, as shown in Figure \ref{fig:doublesphericalpendulum}. As such, the constraints become $\rho_i=l_i$, $i=1,2$, and may be introduced straightforwardly in the Lagrangian. Then, the configuration space is given by 
\begin{equation*}
Q=\mathbb S^2\times\mathbb S^2
\end{equation*}
with angular coordinates $(\theta_1,\varphi_1,\theta_2,\varphi_2)$. After the change of coordinates, the Lagrangian reads
\begin{multline*}
\displaystyle L\left((\theta_1,\varphi_1,\theta_2,\varphi_2),(\dot\theta_1,\dot\varphi_1,\dot\theta_2,\dot\varphi_2)\right)=\vspace{0.1cm}\\
\displaystyle \frac{1}{2}l_1^2\,m_1\,(\dot\varphi_1^2 + \dot\theta_1^2\sin^2\varphi_1)+T_2-g\,m_1\,l_1\cos\varphi_1 - g\,m_2\,(l_1\cos\varphi_1 + l_2\cos\varphi_2).
\end{multline*}
where $T_2$ is the kinetic energy of the second pendulum,
\begin{multline*}
\displaystyle T_2=\frac{1}{2}m_2\,\Big(l_1^2\,\dot\varphi_1^2 + l_2^2\,\dot\varphi_2^2 + l_1^2\,\dot\theta_1^2\sin^2\varphi_1 + l_2^2\,\dot\theta_2^2\sin^2\varphi_2 + 2\,l_1\,l_2\,\big(\dot\varphi_1\,\dot\varphi_2\sin\varphi_1\sin\varphi_2\vspace{0.1cm}\\
\displaystyle + \dot\varphi_1\,\dot\varphi_2\cos\varphi_1\cos\varphi_2\cos(\theta_1 - \theta_2)+ \dot\varphi_1\,\dot\theta_2\sin\varphi_2\sin(\theta_1 - \theta_2)\cos\varphi_1\vspace{0.1cm}\\
\displaystyle - \dot\varphi_2\,\dot\theta_1\sin\varphi_1\sin(\theta_1 - \theta_2)\cos\varphi_2 + \dot\theta_1\,\dot\theta_2\sin\varphi_1\sin\varphi_2\cos(\theta_1 - \theta_2)\big)\Big).
\end{multline*}

Observe that the system is invariant by simultaneous rotation of both pendula around the $Z$-axis, i.e., the group of symmetries is $G=\mathbb S^1$ with the action on $\mathbb S^2\times\mathbb S^2$ given locally by 
\begin{equation*}
\theta\cdot(\theta_1,\varphi_1,\theta_2,\varphi_2)=(\theta+\theta_1,\varphi_1,\theta+\theta_2,\varphi_2),
\end{equation*}
for each $(\theta_1,\varphi_1,\theta_2,\varphi_2)\in\mathbb S^2\times\mathbb S^2$ and $\theta\in\mathbb S^1$. As usual, we denote the quotient by $\Sigma=(\mathbb S^2\times\mathbb S^2)/\mathbb S^1$. 

\begin{remark}
Observe that this action is not free. Indeed, it leaves invariant configurations with $\varphi_1=k_1\pi$ and $\varphi_2=k_2\pi$ for some $k_1,k_2\in\mathbb Z$. Therefore, the following is only valid for trajectories not passing through those configurations. 
\end{remark}

At last, we perform another change of coordinates,
\begin{equation*}
\vartheta_1=\frac{\theta_1+\theta_2}{2},\qquad\vartheta_2=\frac{\theta_2-\theta_1}{2},
\end{equation*}
with $\varphi_1$ and $\varphi_2$ remaining the same. Observe that the inverse is given by $\theta_1=\vartheta_1-\vartheta_2$ and $\theta_2=\vartheta_1+\vartheta_2$. In these coordinates, the action reads
\begin{equation*}
\theta\cdot(\vartheta_1,\varphi_1,\vartheta_2,\varphi_2)=(\theta+\vartheta_1,\varphi_1,\vartheta_2,\varphi_2),
\end{equation*}
for each $(\vartheta_1,\varphi_1,\vartheta_2,\varphi_2)\in\mathbb S^2\times\mathbb S^2$ and $\theta\in\mathbb S^1 $.

\subsubsection{Discrete equations}

For the sake of simplicity, we will let $q_0=(\vartheta_1,\varphi_1,\vartheta_2,\varphi_2)\in\mathbb S^2\times\mathbb S^2$ and $x_0=(\varphi_1,\vartheta_2,\varphi_2)\in(\mathbb S^2\times\mathbb S^2)/\mathbb S^1$, and analogous for $q_0^+$ and $x_0^+$. Given $h> 0$, we define the discrete Lagrangian as
\begin{equation*}
L_d(q_0,q_0^+;h)=hL\left(\frac{q_0+q_0^+}{2},\frac{q_0^+-q_0}{h}\right),\qquad q_0,q_0^+\in\mathbb S^2\times\mathbb S^2.
\end{equation*}
Likewise, we choose the following discrete principal connection
\begin{equation*}
\omega_d\left(q_0,q_0^+\right)=\vartheta_1^+-\vartheta_1,\qquad q_0,q_0^+\in\mathbb S^2\times\mathbb S^2.
\end{equation*}
In particular, the map $\textrm{\normalfont h}_d: Q\times\Sigma\to G$ is given by $\textrm{\normalfont h}_d\left(q_0,x_0^+\right)=\vartheta_1$. Thus, $\textrm{\normalfont h}_{d,\Sigma}\equiv0$, $\textrm{\normalfont h}_{d,\Sigma}^*\equiv0$ and $\textrm{\normalfont h}_d^0\equiv0$. The reduced discrete Lagrangian is
\begin{equation*}
l_d(x_0,x_0^+,\vartheta_1^+;h) = L_d\big((0,\varphi_1,\vartheta_2,\varphi_2),(\vartheta_1^+,\varphi_1^+,\vartheta_2^+,\varphi_2^+);h\big),
\end{equation*}
for each $x_0,x_0^+\in\Sigma$ and $\vartheta_1^+\in\mathbb S^1$. Given the interval $[0,T]\subset\mathbb R$, we divide it into $N=T/h$ subintervals $[t_k,t_{k+1}]$, with $t_k=kh$, $0\leq k\leq N-1$. We denote the discrete vector field by $[\hat X_d^k]=(x_k,w_k,\mu_k,q_{k+1},w_{k+1},\mu_{k+1})$, $0\leq k\leq N-1$. Observe that $w_k=(w_k^1,w_k^2,w_k^3)\in\Sigma^*\simeq\mathbb R^3$ and $\mu_k\in G^*\simeq\mathbb R$. Since $\textrm{\normalfont h}_d((\vartheta_1,0),0)=\vartheta_1$ and $\textrm{\normalfont h}_{d,Q}(x_0,0)=0$, the reduced discrete equations read
\begin{equation*}
\left\{\begin{array}{l}
\displaystyle \left(\frac{\partial l_d}{\partial\varphi_1^+},\frac{\partial l_d}{\partial\vartheta_2^+},\frac{\partial l_d}{\partial\varphi_2^+}\right)=\left(w_1^1,w_1^2,w_1^3\right),\vspace{0.1cm}\\
\displaystyle \mu_0=\mu_1,\vspace{0.1cm}\\
\displaystyle \left(\frac{\partial l_d}{\partial\varphi_1},\frac{\partial l_d}{\partial\vartheta_2},\frac{\partial l_d}{\partial\varphi_2}\right)=-\left(w_0^1,w_0^2,w_0^3\right),\vspace{0.1cm}\\
\displaystyle\frac{\partial l_d}{\partial\vartheta_1^+}=\mu_0,
\end{array}\right.
\end{equation*}
where $(\varphi_1,\vartheta_2,\varphi_2)$, $\mu_0$ and $w_0$ are the initial conditions, and $(\vartheta_1^+,\varphi_1^+,\vartheta_2^+,\varphi_2^+)$, $\mu_1$ and $w_1$ are the unknowns. As in the previous example, by using \eqref{eq:hatlambdadlocal}, the momentum is given by $p_0=(\mu_0,w_0)$.

\begin{figure}[t]
    \centering
    \includegraphics[width=0.6\textwidth]{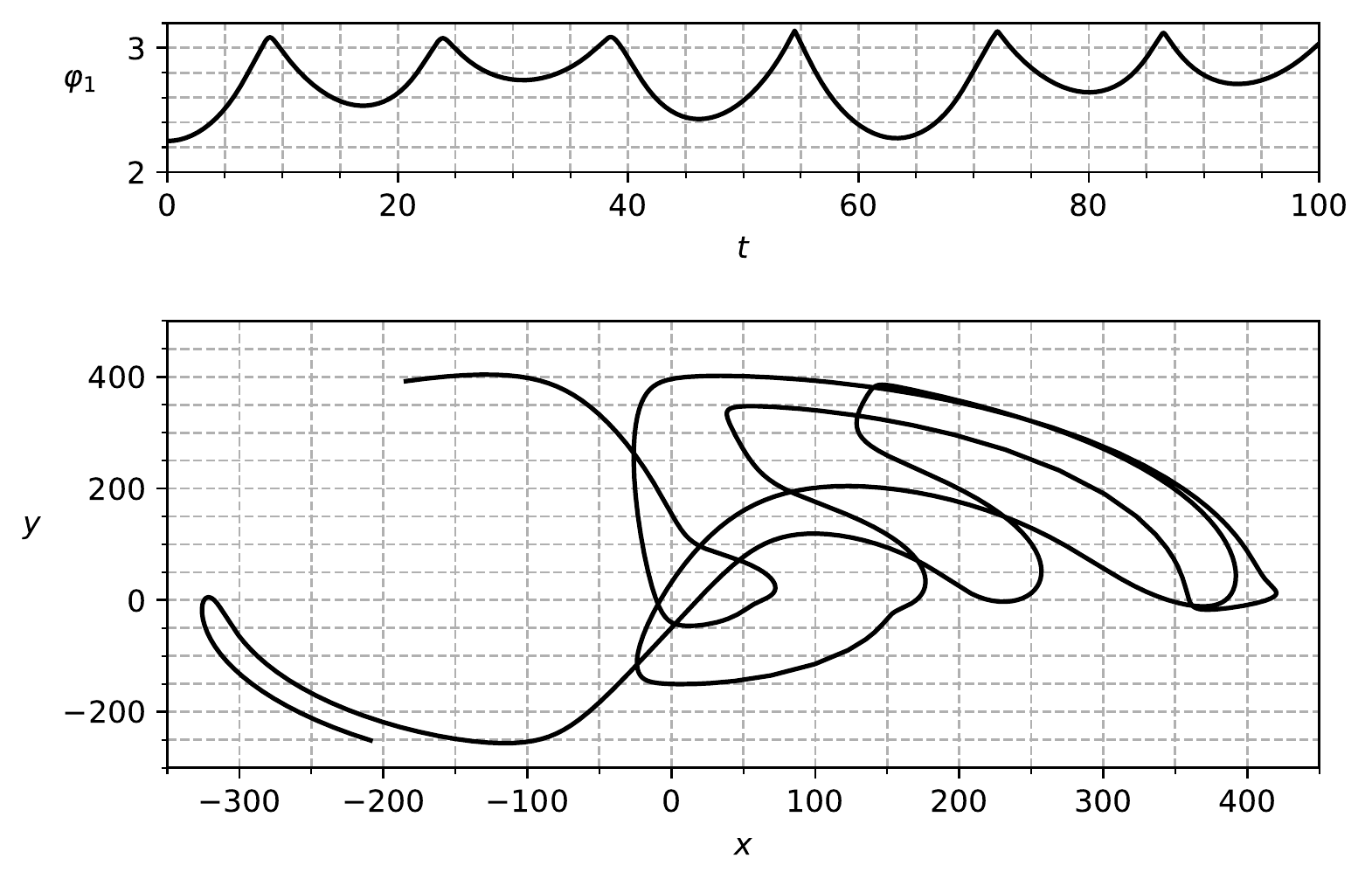}
    \caption{Evolution of $\varphi_1$ (top) and projection of the trajectory of the second pendulum on the $XY$-plane (bottom) for the double spherical pendulum.}
    \label{fig:trajectory}
\end{figure}

\begin{figure}[b]
    \centering
    \includegraphics[width=0.6\textwidth]{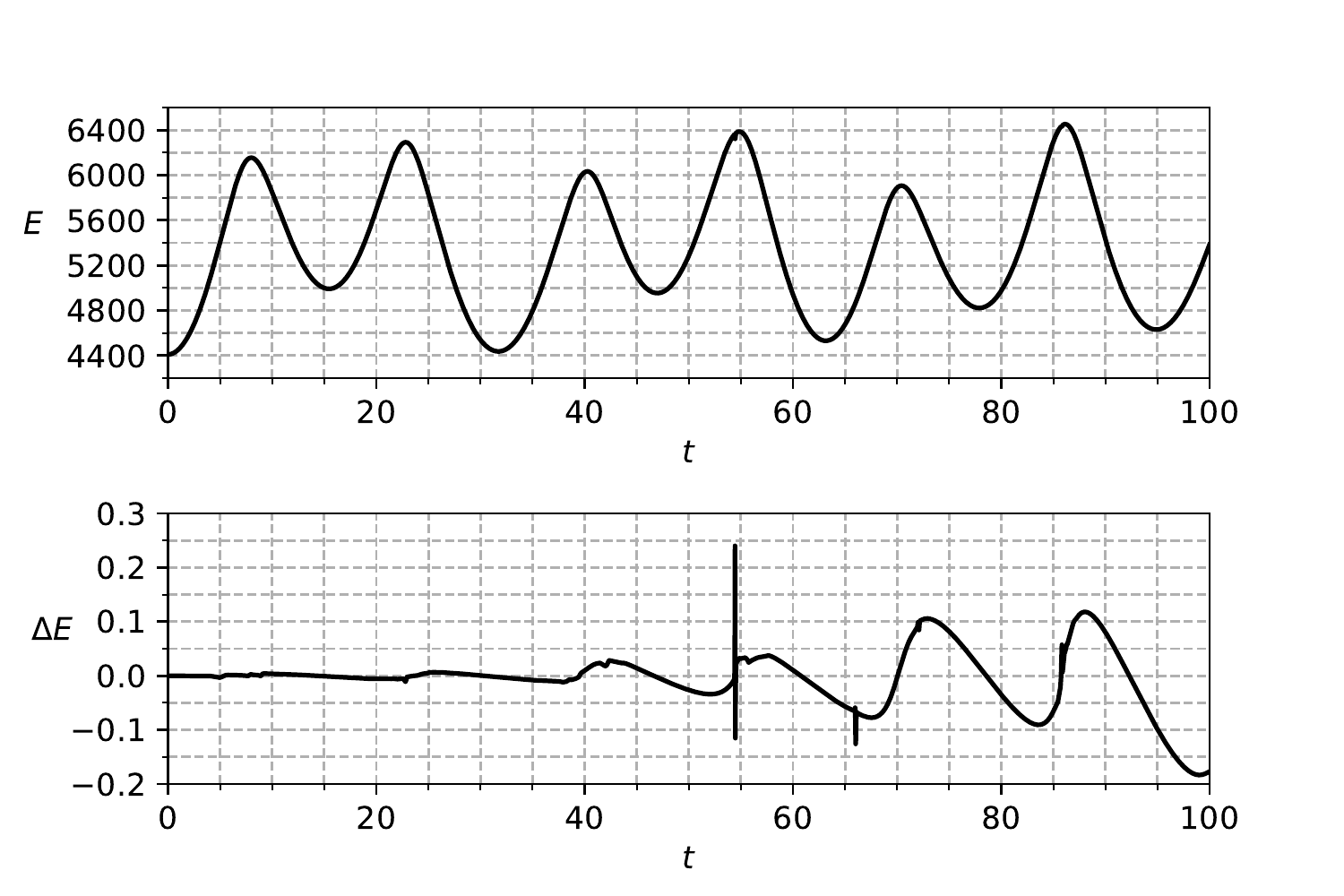}
    \caption{Evolution of the discrete energy for the double spherical pendulum (top) and difference between the energy computed with the discrete reduced equations and the energy computed with the discrete Euler--Lagrange equations (bottom).}
    \label{fig:energy_drift}
\end{figure}

\subsubsection{Numerical computations}

For the numerical simulation, we need to fix the parameters of the system, as well as the initial conditions. Using SI units, we pick $T=100$, $N=10^4$, $m_1=20$, $m_2=35$, $l_1=500$, $l_2=800$, $g=9.8$, $q_0=(0,9/4,2,3)$ and $(\mu_0,w_0)=(0,0,1,1)$. Recall that the $(+)$-discrete generalized energy of the system is given by \eqref{eq:Ed}, i.e., $(E_d)_k=L_d(q_k,q_k^+)$, where we have used that $q_{k+1}=q_k^+$, $0\leq k\leq N-1$. The evolution of the energy is plotted in Figure \ref{fig:energy_drift}. Observe that it exhibits good near energy conservation, since it oscillates around a fixed value instead of exhibiting a spurious drift. This is typical of symplectic, and in particular, variational integrators, but is not generally true of standard integrators. Moreover, we have simulated the system by using the $(+)$-discrete Euler--Lagrange equations. The difference between the energies obtained from our reduced equations and the energy obtained from the Euler--Lagrange equations is also plotted in Figure \ref{fig:energy_drift}. As can be seen, both methods give almost the same values for the discrete energy. However, we have observed that our method was about 75\% faster than the usual discrete Euler--Lagrange equations, which suggests that the reduced integrator could be significantly faster than the unreduced one.

\section{Conclusions}
In this paper, we developed the theory of discrete Dirac reduction of discrete Lagrange--Dirac systems with an abelian symmetry group acting on a linear configuration space. This involves the use of the notion of discrete principal connections to coordinatize the quotient spaces, which allows us to study the reduction of the Dirac structure and the discrete variational principle expressed in terms of the discrete generalized energy. Both the reduced discrete Dirac structure and the reduced discrete variational principle lead to the same reduced discrete equations of motion. We also discussed the role of retractions and the atlas of retraction compatible charts in allowing us to generalize the local theory that was discussed to a discrete reduction theory that is globally well-defined on a manifold.

For future work, we will extend this to the setting of nonabelian symmetry groups, and to discrete analogues of Routh reduction, where the discrete dynamics is also restricted to the level sets of the discrete momentum. In the same vein, it would be interesting to explore discrete Dirac reduction by stages. To that end, a category containing discrete Dirac structures that is closed under quotients must be constructed and the reduction procedure must be defined on the whole category.

\section*{Acknowledgements}

ARA was supported by a FPU grant from the Spanish Ministry of Science, Innovation and Universities (MICIU). ML was supported in part by the NSF under grants DMS-1411792, DMS-1345013, DMS-1813635, CCF-2112665, by AFOSR under grant FA9550-18-1-0288, and by the DoD under grant HQ00342010023 (Newton Award for Transformative Ideas during the COVID-19 Pandemic).

\section*{Declarations}

The authors have no competing interests to declare that are relevant to the content of this article.

\bibliographystyle{plainnat}
\bibliography{biblio.bib}

\end{document}